\documentclass[sigconf]{acmart}

\usepackage{amsmath,amsfonts}
\usepackage{amsthm}
\usepackage{algorithm}
\usepackage[noend]{algorithmic}
\usepackage{graphicx}
\usepackage{textcomp}
\usepackage{url}
\usepackage{xcolor}
\usepackage{bm}
\usepackage{epsfig}
\usepackage{epstopdf}
\usepackage{subfigure}
\usepackage{multirow}
\usepackage{array}

\newtheorem{problem definition}{Problem Definition}

\newtheorem{theorem}{Theorem}

\AtBeginDocument{%
  \providecommand\BibTeX{{%
    \normalfont B\kern-0.5em{\scshape i\kern-0.25em b}\kern-0.8em\TeX}}}

\setcopyright{acmcopyright}
\copyrightyear{2021}
\acmYear{2021}
\acmDOI{10.1145/1122445.1122456}

\acmConference[WWW '21]{WWW '21: The Web Conference}{April 19--23, 2021}{Ljubljana, Slovenia}
\acmBooktitle{WWW '21: The Web Conference,
  April 19--23, 2021, Ljubljana, Slovenia}
\acmPrice{15.00}
\acmISBN{978-1-4503-XXXX-X/18/06}

\begin{document}

\title{Outlier-Resilient Web Service QoS Prediction}

\author{Fanghua Ye}
\affiliation{
\institution{University College London}
\city{London}
\country{UK}
}
\email{fanghua.ye.19@ucl.ac.uk}

\author{Zhiwei Lin}
\affiliation{
\institution{Sun Yat-Sen University}
\city{Guangzhou}
\country{China}
}
\email{linzhw25@mail2.sysu.edu.cn}

\author{Chuan Chen}
\affiliation{
\institution{Sun Yat-Sen University}
\city{Guangzhou}
\country{China}
}
\email{chenchuan@mail.sysu.edu.cn}

\author{Zibin Zheng}
\affiliation{
\institution{Sun Yat-Sen University}
\city{Guangzhou}
\country{China}
}
\email{zhzibin@mail.sysu.edu.cn}

\author{Hong Huang}
\affiliation{
\institution{Huazhong University of Science and Technology, Wuhan, China}
}
\email{honghuang@hust.edu.cn}


\begin{abstract}
  The proliferation of Web services makes it difficult for users to select the most appropriate one among numerous functionally identical or similar service candidates. Quality-of-Service (QoS) describes the non-functional characteristics of Web services, and it has become the key differentiator for service selection. However, users cannot invoke all Web services to obtain the corresponding QoS values due to high time cost and huge resource overhead. Thus, it is essential to predict unknown QoS values. Although various QoS prediction methods have been proposed, few of them have taken outliers into consideration, which may dramatically degrade the prediction performance. To overcome this limitation, we propose an outlier-resilient QoS prediction method in this paper. Our method utilizes Cauchy loss to measure the discrepancy between the observed QoS values and the predicted ones. Owing to the robustness of Cauchy loss, our method is resilient to outliers. We further extend our method to provide time-aware QoS prediction results by taking the temporal information into consideration. Finally, we conduct extensive experiments on both static and dynamic datasets. The results demonstrate that our method is able to  achieve better performance than state-of-the-art baseline methods.
\end{abstract}



\maketitle

\section{Introduction}

Web services provide interoperability among disparate software applications and play a key role in service-oriented computing \cite{bouguettaya2017service}. Over the past few years, numerous Web services have been published as indicated by the Web service repository--ProgrammableWeb\footnote{\url{https://www.programmableweb.com/}}. The proliferation of Web services brings great benefits in building versatile service-oriented applications and systems.

It is apparent that the quality of service-oriented applications and systems relies heavily on the quality of their component Web services. Thus, investigating the quality of Web services is an important task to ensure the reliability of the ultimate applications and the entire systems. The quality of Web services can be characterized by their functional and non-functional attributes. Quality-of-Service (QoS) represents the non-functional aspect of Web services, such as response time, throughput rate and failure probability \cite{papazoglou2007service, wu2015qos}. Since there are many functionally equivalent or similar services offered on the Web, investigating non-functional QoS properties becomes the major concern for service selection \cite{el2010tqos, zhang2011wspred}. However, the QoS value observed by users depends heavily on the Web service invocation context. Hence, the quality of the same Web service experienced by different users may be relatively different \cite{shao2007personalized}. For this reason, it is important to acquire personalized QoS values for different users. Considering that users cannot invoke all Web services to obtain personalized QoS values on their own due to high time cost and huge resource overhead \cite{wu2015qos, zhang2019covering}, predicting missing QoS values based on existing observations plays an essential role in obtaining approximate personalized QoS values.

Matrix factorization (MF) is arguably the most popular technique adopted for QoS prediction \cite{zhang2019covering, zheng2020web, ghafouri2020survey}. However, most existing MF-based QoS prediction methods directly utilize $L_2$-norm to measure the difference between the observed QoS values and the predicted ones \cite{zheng2013personalized, lo2015efficient, xu2016web, xie2016asymmetric, su2016web, wu2018collaborative, wu2019posterior}. It is well-known that $L_2$-norm is sensitive to outliers \cite{zhao2015l1, cao2015low, xu2019adaptive, zhu2017robust}. That is, the objective function value may be dominated by outliers during the $L_2$-norm minimization process, which will lead to severe approximation deviation between the observed normal values and the predicted ones. As a result, without taking outliers into consideration, existing MF-based methods may not achieve satisfactory performance. In recent years, there are some explorations on enhancing the robustness of MF-based QoS prediction methods by replacing $L_2$-norm with $L_1$-norm \cite{zhu2018similarity}. Although $L_1$-norm is more robust to outliers \cite{eriksson2010efficient, zheng2012practical, meng2013cyclic}, $L_1$-norm-based objective function is much harder to optimize and the solution is also unstable \cite{xu2012sparse, meng2013robust}. Moreover, $L_1$-norm is still sensitive to outliers, especially when outliers take significantly different values from the normal ones  \cite{ding2017l1, xu2015multi}. There are also some methods seeking to identify outliers explicitly by means of clustering algorithms \cite{zhu2010ws, wu2019data, hasny2016predicting}, which usually treat all the elements in the smallest cluster as outliers \cite{wu2015qos, he2014location}. However, it is difficult to choose the proper number of clusters. Consequently, these methods usually suffer from the misclassification issue. That is, either some outliers may not be eliminated successfully or some normal values may be selected as outliers falsely.

In this paper, we propose a robust QoS prediction method under the matrix factorization framework to deal with the aforementioned issues. Our method chooses to measure the discrepancy between the observed QoS values and the predicted ones by Cauchy loss \cite{barron2017general, li2018robust} instead of the $L_1$-norm loss or  $L_2$-norm loss. It has been shown that Cauchy loss is much more robust to outliers than the $L_1$-norm loss and $L_2$-norm loss \cite{li2018robust, xu2015multi}. Theoretically, Cauchy loss allows nearly half of the observations to be out of the normal range before it gives incorrect results \cite{mizera2002breakdown}. For a given QoS dataset, it is unlikely that nearly half of the observations are outliers. Thus, Cauchy loss is sufficient for outlier modeling and has the potential to provide better prediction results. Note also that our method does not explicitly identify outliers, which reduces the risk of misclassification and makes our method more general and more robust. In other words, our method is resilient to outliers. Considering that the QoS value of a Web service observed by a particular user may change over time, it is essential to provide time-aware personalized QoS prediction results. To achieve this goal, we further extend our method under the tensor factorization framework by taking the temporal information into consideration.

In summary, the main contributions of this paper include:
\begin{itemize}
	\item First, we propose a robust Web service QoS prediction method with outlier resilience. Our method measures the discrepancy between the observed QoS values and the predicted ones by Cauchy loss, which is robust to outliers.
	
	\item Second, we extend our method to provide time-aware QoS prediction results under the tensor factorization framework by taking the temporal information into consideration.
	
	\item Third, we conduct extensive experiments on both static and dynamic datasets to evaluate the performance of our method. The results demonstrate that our method can achieve better performance than state-of-the-art baseline methods.
\end{itemize}


\section{Unavoidability of QoS Outliers} \label{sec:mot}

Most existing QoS prediction methods assume that the QoS observations are reliable and rational. However, this assumption may not hold in the real world. This is because the observed QoS data can be affected by many factors. For example, there may be some malicious users submitting wrong QoS values deliberately. The service providers may also pretend to be service users and thus exaggerate the performance of their own Web services and depreciate the performance of their competitors' Web services. In addition, the QoS values observed by users are largely dependent on the invocation environment such as network latency and server overload, which may lead some of the QoS values to  deviate far from the normal range. In consideration of these complicated factors, we argue that it is highly possible that some of the QoS observations are outliers.

However, we are in lack of an oracle showing in advance which QoS observations are outliers. Here, we treat the rare extreme values which significantly differ from the remaining ones as outliers by following the definition in \cite{kim2016outlier}. 
To be more intuitive, in Figure~\ref{fig:diff1} (a) and (b), we show the distribution of both response time and throughput of 100 Web services invoked by 3 randomly selected users from a publicly available dataset--WS-DREAM\_dataset1\footnote{\url{https://github.com/wsdream/wsdream-dataset/tree/master/dataset1}}. As can be seen, although a user tends to have different usage experiences on different Web services, most QoS values of these Web services observed by the three users fall into a normal range. For example, the response time mainly falls in the interval of $[0, 2]$. However, there are also some observations deviating far from the normal range. As shown in Figure~\ref{fig:diff1} (a), the response time experienced by user 2 even reaches up to 14 seconds, which is far beyond the normal range. Needless to say, such kind of observations should be treated as outliers. We further demonstrate the distribution of response time and throughput of 3 Web services experienced by 100 different users in Figure~\ref{fig:diff1} (c) and (d). It can be observed that although the usage experiences of a Web service can vary widely among different users, the QoS values of the same Web service observed by the majority of users tend to fall into a normal range. Whereas, there are also some observations taking values far beyond the normal range. These phenomena verify the rationality of treating extreme values as outliers and also reveal the unavoidability of outliers in QoS observations.

\begin{figure}[t!]
	\centering
	\subfigure[{Response Time}]{
		\includegraphics[width=0.47\columnwidth]{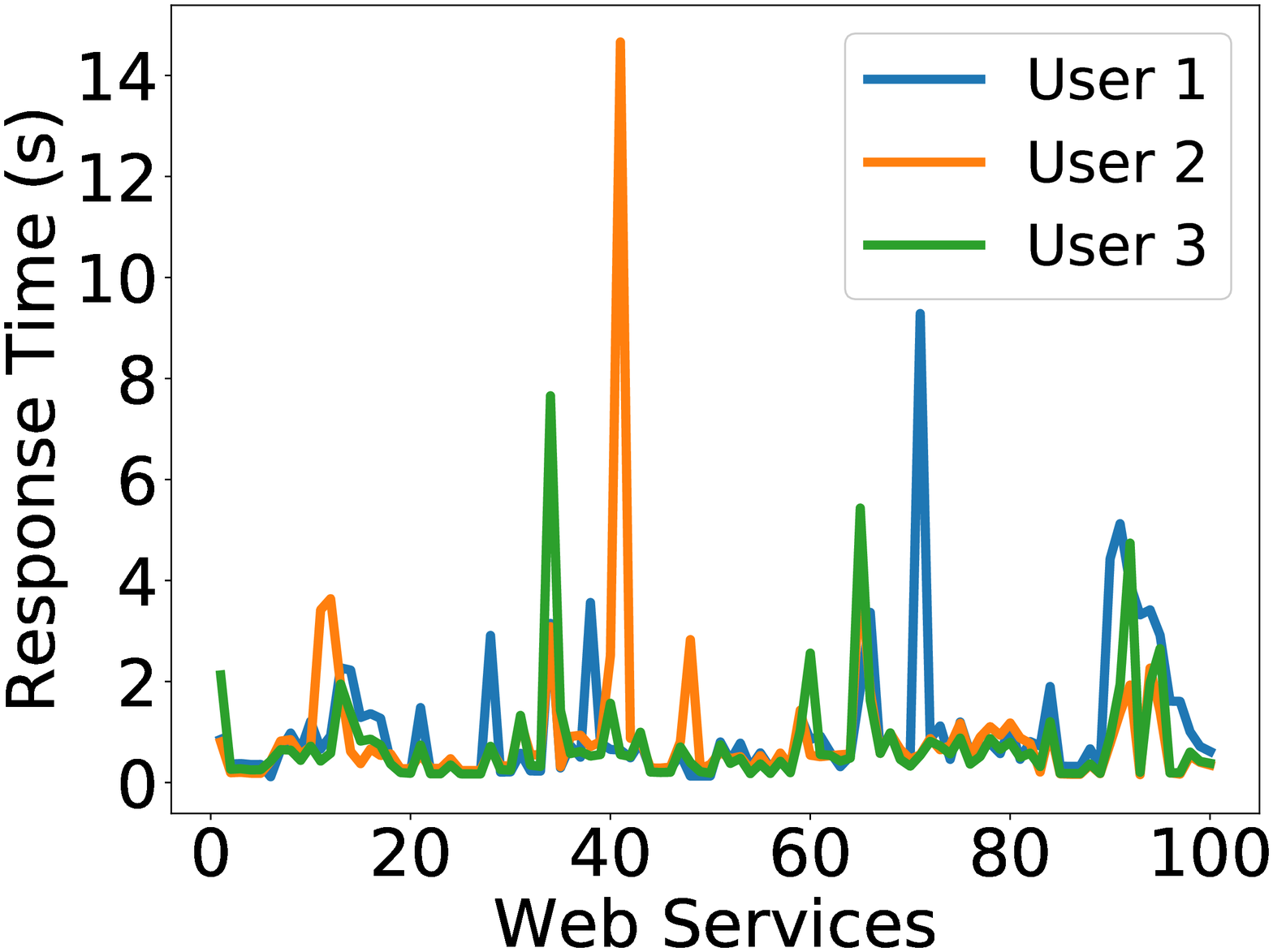}
	}
	\subfigure[{Throughput}]{
		\includegraphics[width=0.4835\columnwidth]{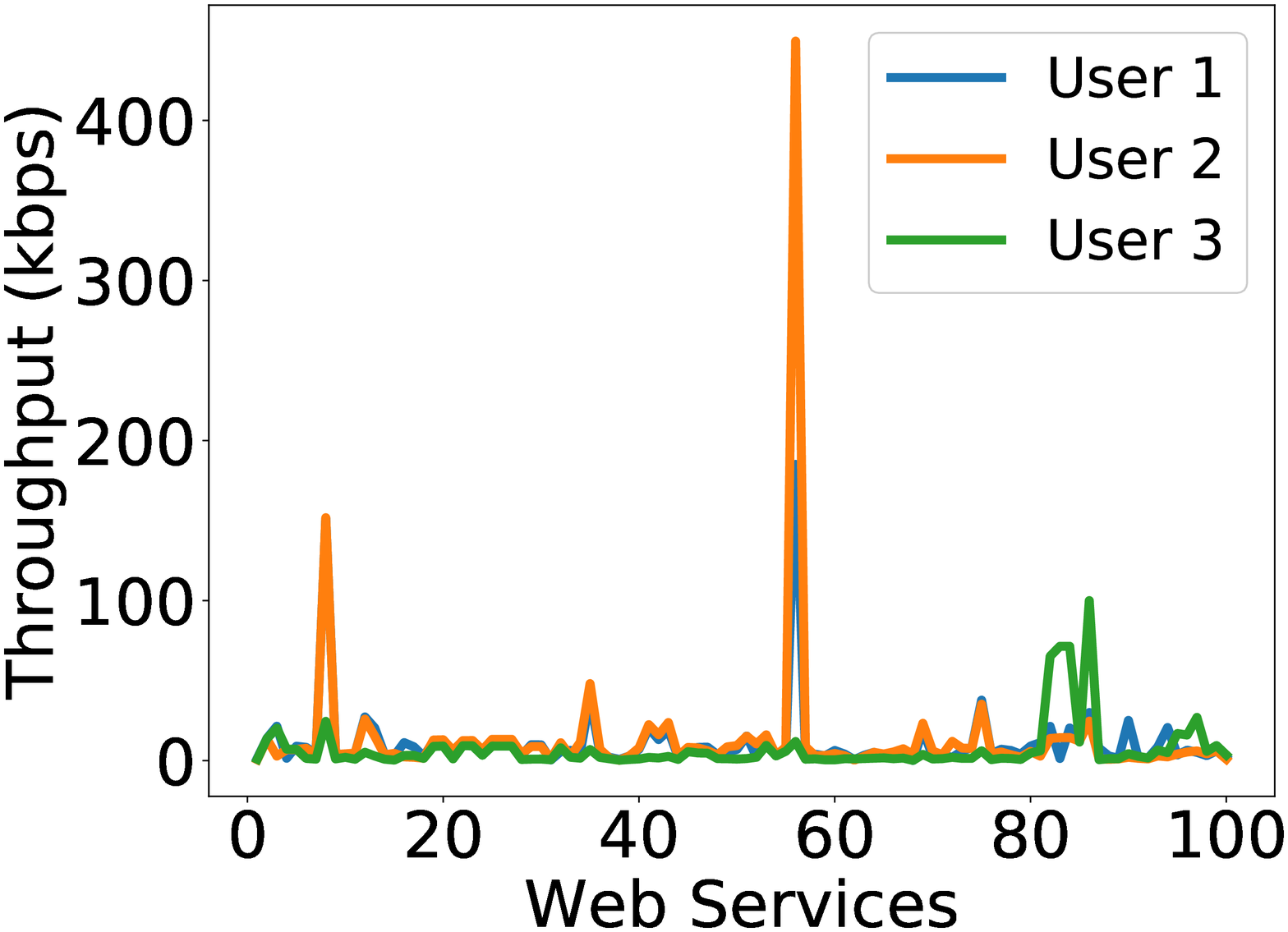}
	}
    \subfigure[{Response Time}]{
		\includegraphics[width=0.47\columnwidth]{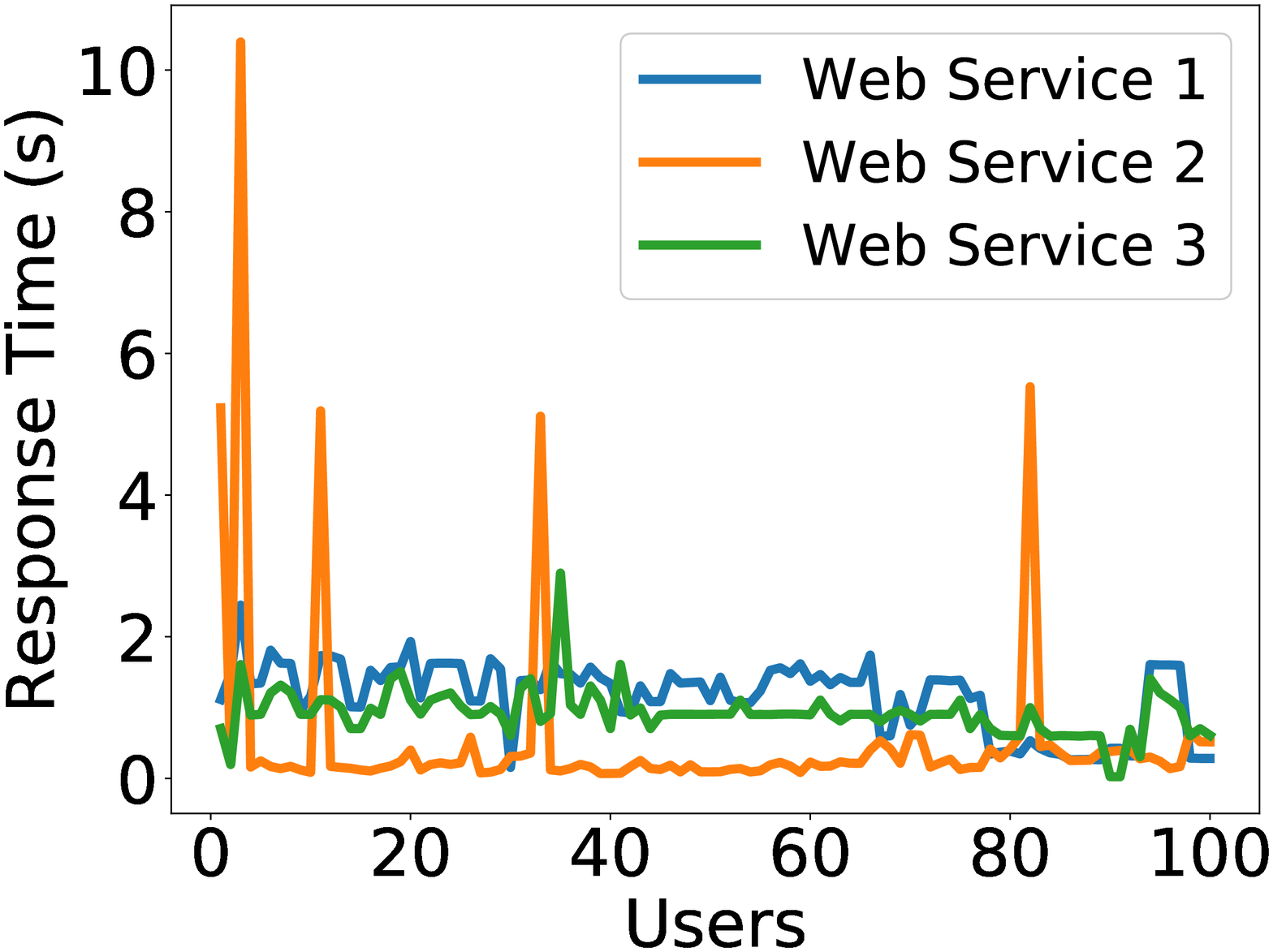}
	}
	\subfigure[{Throughput}]{
		\includegraphics[width=0.4835\columnwidth]{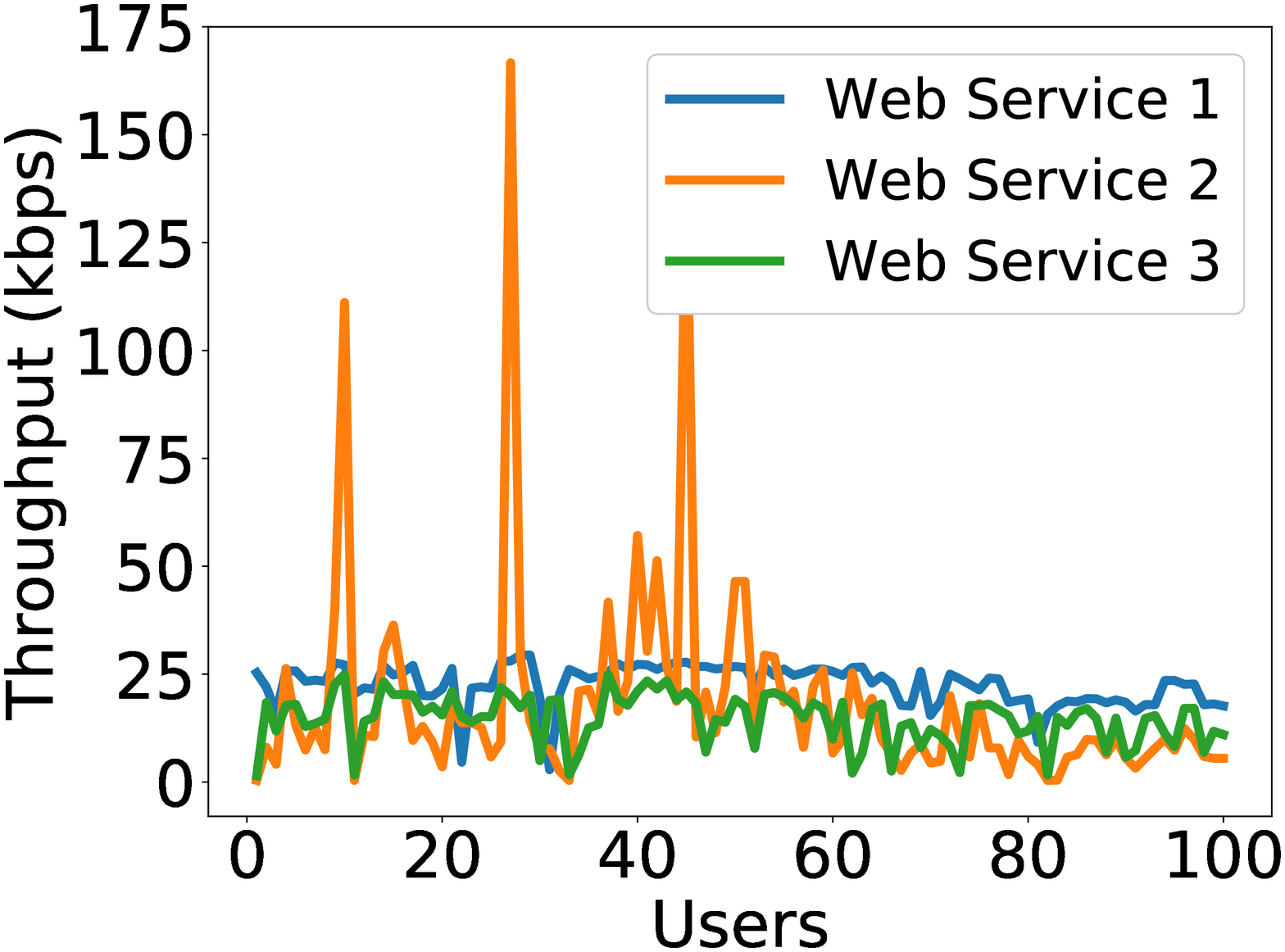}
	}
	\vspace*{-0.3cm}
	
	\caption{The distribution of response time and throughput.}
	\label{fig:diff1}
\end{figure}

\section{Preliminaries} \label{sec:pre}

Suppose that we are provided with a set of $m$ users and a set of $n$ Web services, then the QoS values between all users and Web services can be represented by a matrix ${\bm X} \in \mathbb{R}^{m \times n}$ whose entry ${\bm X}_{ij}$ denotes the QoS value of Web service $j$ observed by user $i$. Obviously, it is time-consuming and resource-consuming for each user to invoke all Web services to get the personalized QoS values. As a consequence, we typically have only partial observations between users and Web services, which means that lots of entries in ${\bm X}$ are null. The goal of QoS prediction is to predict these null entries by exploiting the information contained in existing observations.

\subsection{Problem Definition}

Let $\Omega$ denote the set of existing QoS observations, that is,
\begin{equation}
\begin{split}
\Omega = \{(i, j, {\bm X}_{ij}) ~\vert~ & \text{the QoS value ${\bm X}_{ij}$ between user $i$ and} \\
	     &\text{ Web service $j$ has been observed} \}.
\end{split}
\end{equation}
Then the problem of QoS prediction is defined as follows.

\textbf{Problem Statement:} Given a set of QoS observations $\Omega$, QoS prediction aims at predicting unknown QoS values by utilizing the information contained in $\Omega$.

\subsection{Matrix Factorization for QoS Prediction}

Generally speaking, matrix factorization tries to factorize a given matrix into the product of several low-rank factor matrices. In the context of QoS prediction, the basic framework of MF-based methods is to factorize matrix ${\bm X}$ into two low-rank factor matrices ${\bm U} \in \mathbb{R}^{m \times l}$ and ${\bm S} \in \mathbb{R}^{n \times l}$, i.e., ${\bm X} \approx {\bm U}{\bm S}^T$. Here, each row of ${\bm U}$ represents the latent feature of a user, and each row of ${\bm S}$ represents the latent feature of a Web service. The dimensionality of latent features is controlled by a parameter $l$ $(l \ll \min(m,n))$. Apparently, ${\bm U}{\bm S}^T$ should be as close to ${\bm X}$ as possible. As thus, the general objective function for MF-based QoS prediction methods can be derived as:
\begin{equation}
\min_{{\bm U}, {\bm S}} \mathcal{L}({\bm X}, {\bm U}{\bm S}^T) + \lambda \mathcal{L}_{reg},
\end{equation}
where $\mathcal{L}$ measures the degree of approximation between ${\bm U}{\bm S}^T$ and ${\bm X}$, $\mathcal{L}_{reg}$ denotes the regularization term to avoid over-fitting, and $\lambda$ represents the regularization coefficient.

The most widely adopted loss function in matrix factorization is the least square loss (i.e., $L_2$-norm loss), which is also the most commonly used loss function in MF-based QoS prediction methods \cite{zheng2013personalized, lo2015efficient, xu2016web, zhang2019covering, zhu2018similarity}. In this setting, the specific objective function can be clearly given as follows:
\begin{equation} \label{eqn:l2}
\min_{{\bm U}, {\bm S}} \frac{1}{2} \Vert {\bm I} \odot ({\bm X} -  {\bm U}{\bm S}^T) \Vert^2_2 + \lambda \mathcal{L}_{reg},
\end{equation}
where $\Vert \cdot \Vert_2$ denotes the $L_2$-norm which is calculated as the square root of the sum of squares of all entries, $\odot$ denotes the Hadamard product (i.e., entry-wise product), and ${\bm I} \in \mathbb{R}^{m \times n}$ denotes the indicator matrix whose entry ${\bm I}_{ij}$ indicates whether the QoS value of Web service $j$ has been observed by user $i$ or not. If user $i$ has the record of Web service $j$, ${\bm I}_{ij}$ is set to 1; otherwise, it is set to 0.


The objective function based on $L_2$-norm as in Eq.~\eqref{eqn:l2} is smooth and  can be optimized by the gradient descent method \cite{gemulla2011large}. However, the $L_2$-norm is sensitive to outliers (i.e., rare extreme values) \cite{zhao2015l1}. When the given observations contain outliers, the residuals between these outliers' corresponding entries in ${\bm X}$ and their approximation entries in ${\bm U}{\bm S}^T$ become huge due to the square operation. Therefore, when minimizing the objective function in Eq.~\eqref{eqn:l2}, more priorities are  given to these outliers, which unfortunately causes severe approximation deviation of the normal QoS values. As a result, the QoS prediction performance may degrade dramatically.

To make the model more robust to outliers, a common stategy is to replace $L_2$-norm with $L_1$-norm \cite{zhu2018similarity, wu2018multiple, ke2005robust, eriksson2010efficient}. Based on $L_1$-norm, the objective function is formularized as below:
\begin{equation} \label{eqn:l1}
\min_{{\bm U}, {\bm S}} \Vert {\bm I} \odot ({\bm X} -  {\bm U}{\bm S}^T) \Vert_1 + \lambda \mathcal{L}_{reg},
\end{equation}
where $\Vert \cdot \Vert_1$ denotes the $L_1$-norm which is calculated as the sum of the absolute values of all entries. Although $L_1$-norm is to some extent more robust to outliers than $L_2$-norm, the objective function based on $L_1$-norm as in Eq.~\eqref{eqn:l1} is a non-smooth function and it is much harder to optimize. What's more, although the large residuals due to outliers are not squared in $L_1$-norm, they may still be quite large relative to the normal ones and thus one would expect that they would influence the objective function as well \cite{ding2017l1}.

\section{Our Method} \label{sec:method}

As stated in the previous section, both $L_1$-norm and $L_2$-norm are sensitive to outliers. In order to make the MF-based methods more robust to outliers, we propose a novel QoS prediction method that utilizes Cauchy loss \cite{barron2017general} as the measurement of the discrepancy between the observed QoS values and the predicted ones. It has been shown that Cauchy loss is resistant to outliers \cite{xu2015multi, guan2017truncated, mizera2002breakdown}. Thus our method is expected to be robust to outliers.

\subsection{M-Estimator}

Before presenting the details of our method, we first introduce the concept of M-estimator. In robust statistics, M-estimators are a broad class of estimators, which represent the minima of particular loss functions \cite{huber2011robust}. Let $r_i$ denote the residual of the $i$-th datum, i.e., the difference between the $i$-th observation and its approximation. Then M-estimators try to optimize the following objective function:
\begin{equation}
\min \sum_{i} g(r_i),
\end{equation}
where function $g$ gives the contribution of each residual to the objective function. A reasonable function $g$ should satisfy the following four properties \cite{fox2002r}:
\begin{itemize}
	\item $g(x) \geq 0$, $\forall x$;
	\item $g(x) = g(-x)$, $\forall x$;
	\item $g(0) = 0$;
	\item $g(x)$ is non-decreasing in $|x|$, i.e., $g(x_1) \leq g(x_2)$, $ \forall |x_1| < |x_2|$.
\end{itemize}
The influence function of $g$ is defined as its first-order derivative: 
\begin{equation}
g'(x) = \frac{\mathrm{d}{g(x)}}{\mathrm{d}{x}}.
\end{equation}
The influence function $g'$ measures the influence of each datum on the value of the parameter estimate. For a robust M-estimator, it would be inferred that the influence of any single datum is insufficient to yield any significant offset \cite{xu2015multi}. Ideally, a robust M-estimator should have a bounded influence function.

Both $L_2$-norm loss and $L_1$-norm loss satisfy the four properties required by M-estimators. For the $L_2$ estimator with $g(x) = \frac{1}{2} x^2$, the influence function is $g'(x) = x$, which means that the influence of a datum on the parameter estimate grows linearly as the error increases. This confirms the non-robusteness of $L_2$ estimator to outliers. Although the $L_1$ estimator with $g(x) = |x|$ can reduce the influence of large errors due to its bounded influence function, it will still be affected by outliers since its influence function has no cut off point \cite{xu2015multi, li2018robust} ($\vert g'(x)\vert = 1$ even when $x \rightarrow \pm \infty$). Besides, $L_1$ estimator is not stable because $g(x) = |x|$ is not strictly convex in $x$. It follows that the influence function of a robust M-estimator should not only be bounded but also be insensitive to the increase of errors ($\vert g'(x)\vert \rightarrow 0$  when $x \rightarrow \pm \infty$). Cauchy estimator has been shown to possess such precious characteristics. The $g$ function of Cauchy estimator (i.e., Cauchy loss) is shown as follows:
\begin{equation}
g(x) = \ln \left(1 + \frac{x^2}{\gamma^2} \right),
\end{equation}
where $\gamma$ is a constant. The influence function is then calculated as:
\begin{equation}
g'(x) = \frac{2x}{\gamma^2 + x^2},
\end{equation}
which takes value in the range of $[-\frac{1}{\gamma}, \frac{1}{\gamma}]$. Moreover, $g'(x)$ tends to zero when $x$ goes to infinity. This indicates that the influence function of Cauchy estimator is insensitive to the increase of errors. Therefore, Cauchy estimator is robust to outliers. A comparison of different M-estimators is illustrated in Figure~\ref{fig:diff}.

\begin{figure}[t!]
	\centering
	\subfigure[{The Loss Function}]{
		\includegraphics[width=0.4743\columnwidth]{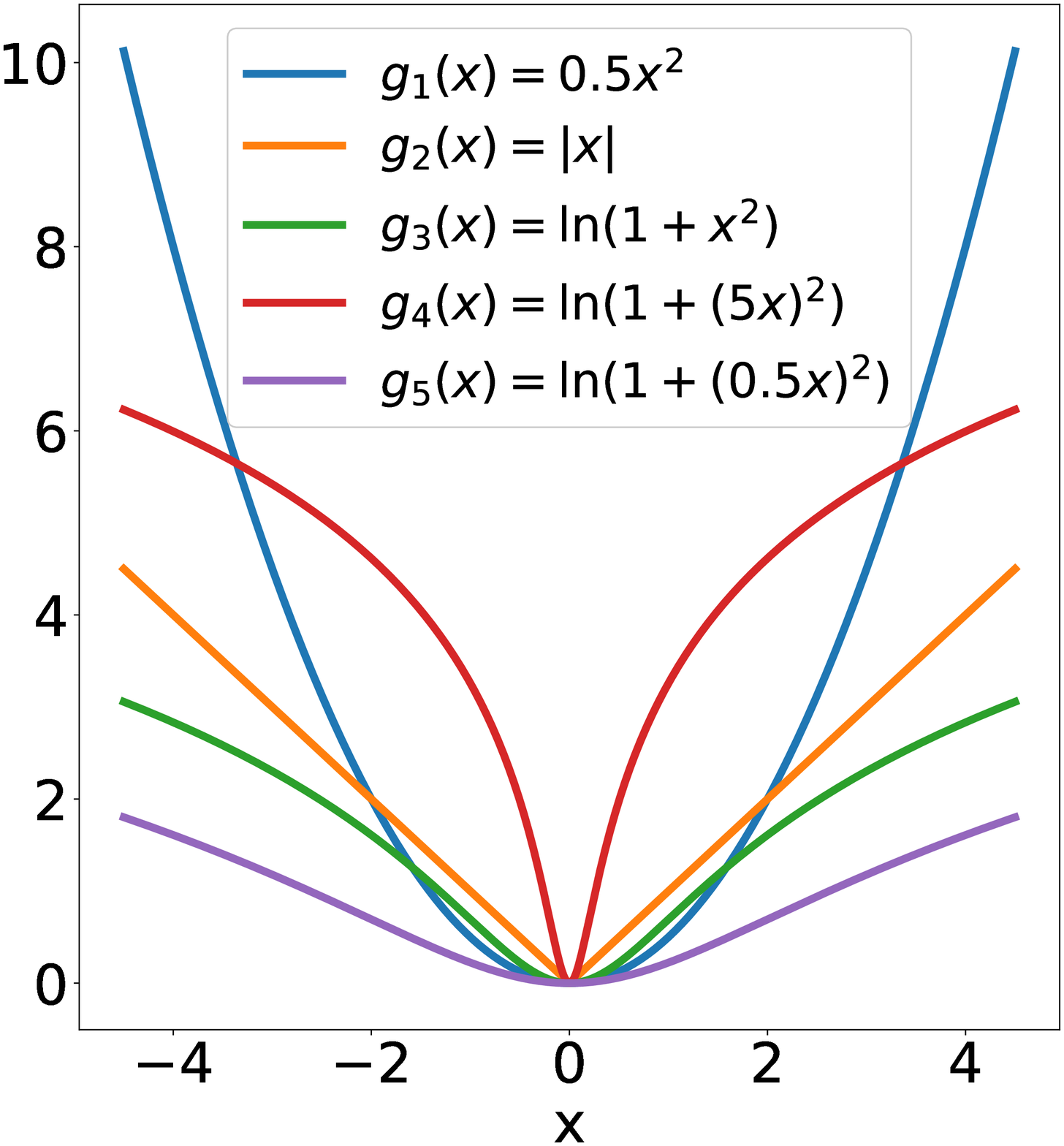}
	}
	\subfigure[{The Influence Function}]{
		\includegraphics[width=0.4785\columnwidth]{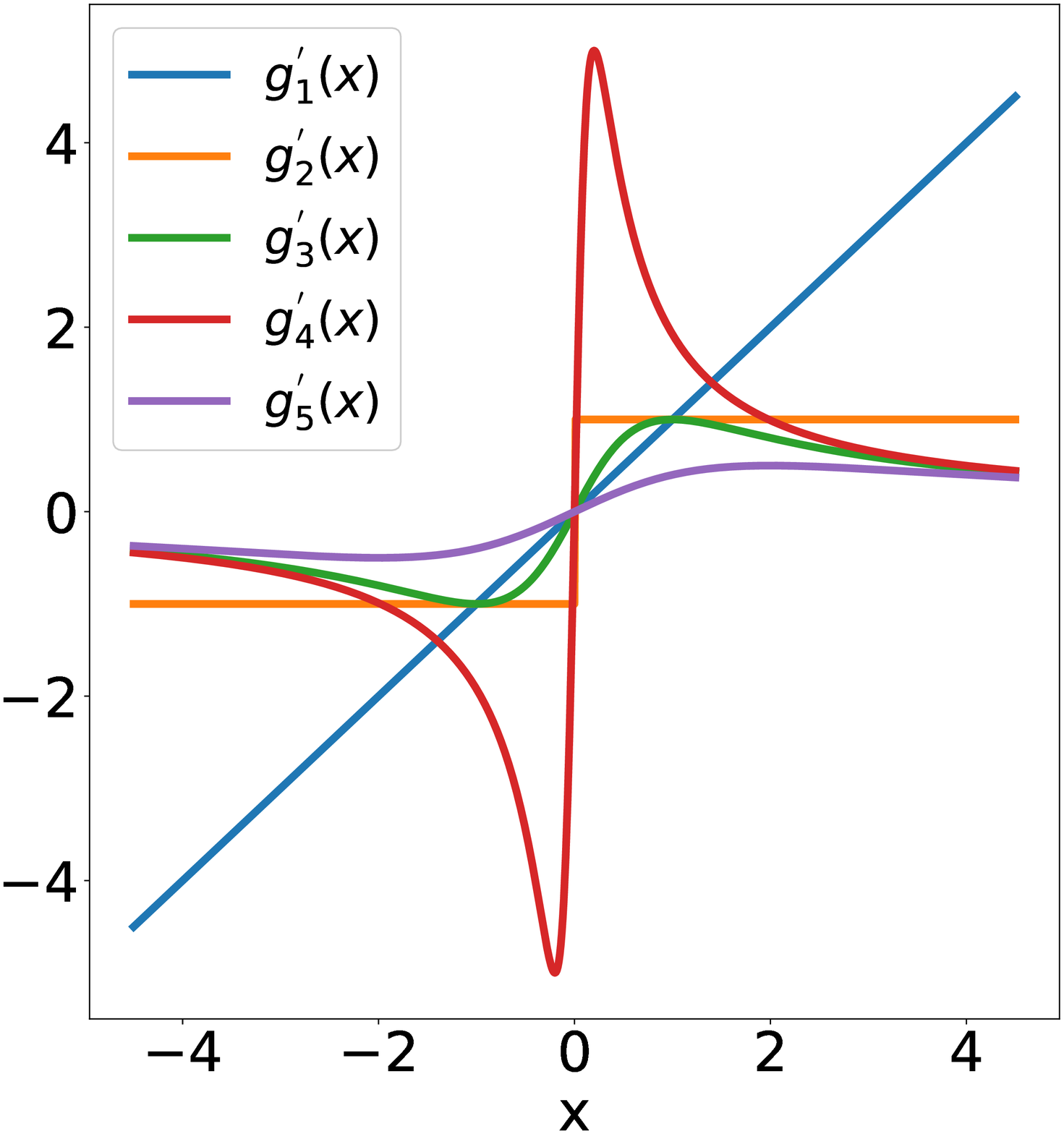}
	}
	\caption{Comparison of different M-estimators.}
	\label{fig:diff}
\end{figure}

\subsection{Model Formulation}

In view of Cauchy estimator's robustness, we choose Cauchy loss to construct the objective function of our method. Based on Cauchy loss, the objective function is derived as:
\begin{equation} \label{eqn:losscmf}
\begin{split}
\min_{{\bm U}, {\bm S}} \mathcal{L} =  \frac{1}{2} \sum_{i=1}^m \sum_{j=1}^n {\bm I}_{ij} \ln \left(1 + \frac{({\bm X}_{ij} - {\bm U}_i {\bm S}^T_j)^2}{\gamma^2} \right) 
           + \frac{\lambda_u}{2} \Vert {\bm U} \Vert^2_2 + \frac{\lambda_s}{2} \Vert {\bm S} \Vert^2_2,
\end{split}
\end{equation}
where ${\bm U}_i$ and ${\bm S}_j$ denote the $i$-th row of ${\bm U}$ and the $j$-th row of ${\bm S}$ respectively, $\lambda_u$ and  $\lambda_s$ represent the regularization coefficients.

The objective function in Eq.~\eqref{eqn:losscmf} can be efficiently optimized by the gradient descent method \cite{gemulla2011large}. Specifically, we choose to optimize ${\bm U}$ and ${\bm S}$ row by row. Then, we have the following update rules:
\begin{equation} \label{eqn:u1}
{\bm U}_i \leftarrow {\bm U}_i - \eta_u \frac{\partial{\mathcal{L}}}{\partial{{\bm U}_i}},
\end{equation}
\begin{equation} \label{eqn:s1}
{\bm S}_j \leftarrow {\bm S}_j - \eta_s \frac{\partial{\mathcal{L}}}{\partial{{\bm S}_j}},
\end{equation}
where $\eta_u$ and $\eta_s$ denote the learning rates for ${\bm U}$ and ${\bm S}$, and
\begin{equation}
\frac{\partial{\mathcal{L}}}{\partial{{\bm U}_i}} = \lambda_u {\bm U}_i - \sum_{j=1}^n {\bm I}_{ij} \frac{{\bm X}_{ij} - {\bm U}_i {\bm S}^T_j}{\gamma^2 + ({\bm X}_{ij} - {\bm U}_i {\bm S}^T_j)^2} {\bm S}_j,
\end{equation}
\begin{equation}
\frac{\partial{\mathcal{L}}}{\partial{{\bm S}_j}} = \lambda_s {\bm S}_j - \sum_{i=1}^m {\bm I}_{ij} \frac{{\bm X}_{ij} - {\bm U}_i {\bm S}^T_j}{\gamma^2 + ({\bm X}_{ij} - {\bm U}_i {\bm S}^T_j)^2} {\bm U}_i.
\end{equation}

The overall optimization procedure of our method is presented in Algorithm~\ref{alg:aaa}, whose time complexity is shown in Theorem~\ref{the:static}. 
\begin{theorem} \label{the:static}
	Let $r$ denote the number of iterations for Algorithm~\ref{alg:aaa} to achieve convergence and let $\rho$ denote the number of available entries in ${\bm X}$, then the time complexity of Algorithm~\ref{alg:aaa} is $\mathcal{O}(r \rho l)$.
\end{theorem}
\begin{proof}
	The main time cost of Algorithm~\ref{alg:aaa} lies in the updates of ${\bm U}$ and ${\bm S}$. In each iteration, updating ${\bm U}$ takes $\mathcal{O}(ml + \rho l)$ time and updating ${\bm S}$ takes $\mathcal{O}(nl + \rho l)$ time. Since both $m$ and $n$ are less than $\rho$, the time complexity of updating ${\bm U}$ and ${\bm S}$ can both be simplified as $\mathcal{O}(\rho l)$. Thus, the overall time complexity is of order $\mathcal{O}(r \rho l)$.
\end{proof}

\begin{algorithm}[t!]
	\caption{Algorithm for Static QoS Prediction}
	\label{alg:aaa}
	\begin{algorithmic}[1]
		\REQUIRE ${\bm X} \in \mathbb{R}^{m \times n}$, $l$, $\gamma$, $\lambda_u$, $\lambda_s$, $\eta_u$, $\eta_s$;
		\ENSURE ${\bm U} \in \mathbb{R}^{m \times l}$, ${\bm S} \in \mathbb{R}^{n \times l}$;
		\STATE Randomly initialize ${\bm U}$ and ${\bm S}$;
		\REPEAT
		\FOR {$i = 1$ to $m$}
		    \STATE Update ${\bm U}_i$ according to Eq.~\eqref{eqn:u1};
		\ENDFOR
		\FOR {$j = 1$ to $n$}
		\STATE Update ${\bm S}_j$ according to Eq.~\eqref{eqn:s1};
		\ENDFOR
		\UNTIL{Convergence}
		\RETURN ${\bm U}$, ${\bm S}$;
	\end{algorithmic}
\end{algorithm}

\subsection{Extension for Time-Aware QoS Prediction}

As pointed out in \cite{zhang2011wspred}, the QoS performance of Web services is highly related to the invocation time because the service status (e.g., number of users) and the network environment (e.g., network speed) may change over time. Thus, it is essential to provide time-aware personalized QoS information to help users make service selection at runtime \cite{zhu2017online}. In this part, we aim to extend our method to make it suitable for time-aware personalized QoS prediction.

To achieve the goal, we choose to extend our method under the tensor factorization framework. A tensor is a multidimensional or $N$-way array \cite{kolda2009tensor}. An $N$-way tensor is denoted as $\bm{\mathcal{X}} \in \mathbb{R}^{I_1 \times I_2 \times \cdots \times I_N}$, which has $N$ indices ($i_1, i_2, \cdots, i_N$)
and its entries are denoted by $\bm{\mathcal{X}}_{i_1 i_2 \cdots i_N}$. In this sense, a tensor can be treated as a generalized matrix and a matrix can also be treated as a two-way tensor.

For time-aware QoS prediction, we need to take the temporal information of QoS values into consideration. According to the definition of tensors, it is clear that we can model QoS observations with temporal information as a three-way tensor $\bm{\mathcal{X}} \in \mathbb{R}^{m \times n \times t}$ whose each entry $\bm{\mathcal{X}}_{ijk}$ represents the QoS value of Web service $j$ observed by user $i$ at time $k$. Here, $t$ denotes the total number of time intervals. Accordingly, we can use an indicator tensor   $\bm{\mathcal{I}} \in \mathbb{R}^{m \times n \times t}$ to show whether the QoS values have been observed or not. If user $i$ has the record of Web service $j$ at time $k$, $\bm{\mathcal{I}}_{ijk}$ is set to 1; otherwise, its value is set to 0. To predict the unknown QoS values in $\bm{\mathcal{X}}$, similar to MF-based methods, we first factorize $\bm{\mathcal{X}}$ to learn the latent features of users, Web services and contexts, respectively, and then leverage the interaction among them to predict QoS values. Specifically,  we adopt the canonical polyadic (CP) decomposition method \cite{rabanser2017introduction} to factorize $\bm{\mathcal{X}}$ into three low-rank factor matrices $\bm{U} \in \mathbb{R}^{m \times l}$, $\bm{S} \in \mathbb{R}^{n \times l}$ and $\bm{T} \in \mathbb{R}^{t \times l}$. Then, $\bm{\mathcal{X}}$ is approximated in the following way:
\begin{equation}
\bm{\mathcal{X}} \approx  \bm{\mathcal{\hat X}} = \sum_{\ell = 1}^l \bm{U}^{(\ell)} \circ \bm{S}^{(\ell)} \circ \bm{T}^{(\ell)},
\end{equation}
where $\bm{U}^{(\ell)} \in \mathbb{R}^{m}$, $\bm{S}^{(\ell)} \in \mathbb{R}^{n}$ and $\bm{T}^{(\ell)} \in \mathbb{R}^{t}$ denote the $\ell$-th column of  $\bm{U}$, $\bm{S}$ and $\bm{T}$ respectively, and $\circ$ represents the vector outer product. In this way, each entry $\bm{\mathcal{X}}_{ijk}$ is approximated by:
\begin{equation} \label{eqn:pair}
\begin{split}
\bm{\mathcal{X}}_{ijk} \approx  \bm{\mathcal{\hat X}}_{ijk} &= \sum_{\ell = 1}^l \bm{U}_{i\ell} \bm{S}_{j\ell} \bm{T}_{k\ell}.
\end{split}
\end{equation}

For tensor factorization, nonnegative constraints are usually enforced on the factor matrices to promote the model's interpretability \cite{shashua2005non}. We also add nonnegative constraints to all the factor matrices ${\bm U}$, ${\bm S}$ and ${\bm T}$. Then together with the Cauchy loss, we derive the objective function for time-aware QoS prediction as below:
\begin{equation} \label{eqn:lossctf}
\begin{split}
\min_{{\bm U}, {\bm S}, {\bm T}} & \mathcal{L}' =  \frac{1}{2} \sum_{i=1}^m \sum_{j=1}^n \sum_{k=1}^t \bm{\mathcal{I}}_{ijk} \ln \left(1 + \frac{(\bm{\mathcal{X}}_{ijk} - \bm{\mathcal{\hat X}}_{ijk})^2}{\gamma^2} \right) \\
& ~ \quad \quad  + \frac{\lambda_u}{2} \Vert {\bm U} \Vert^2_2 + \frac{\lambda_s}{2} \Vert {\bm S} \Vert^2_2 + \frac{\lambda_t}{2} \Vert {\bm T} \Vert^2_2, \\
s.t. ~& {\bm U} \geq {\bm 0}, {\bm S} \geq {\bm 0}, {\bm T} \geq {\bm 0},
\end{split}
\end{equation}
where $\lambda_t$ denotes the regularization coefficient for matrix $\bm{T}$.

\begin{algorithm}[t!]
	\caption{Algorithm for Time-Aware QoS Prediction}
	\label{alg:bbb}
	\begin{algorithmic}[1]
		\REQUIRE $\bm{\mathcal{X}} \in \mathbb{R}^{m \times n \times t}$, $l$, $\gamma$, $\lambda_u$, $\lambda_s$, $\lambda_t$;
		\ENSURE ${\bm U} \in \mathbb{R}^{m \times l}$, ${\bm S} \in \mathbb{R}^{n \times l}$,  ${\bm T} \in \mathbb{R}^{t \times l}$;
		\STATE Randomly initialize ${\bm U} \geq {\bm 0}$, ${\bm S} \geq {\bm 0}$ and ${\bm T} \geq {\bm 0}$;
		\REPEAT
		\FOR {$i = 1$ to $m$}
		\STATE Update ${\bm U}_i$ according to Eq.~\eqref{eqn:u2};
		\ENDFOR
		\FOR {$j = 1$ to $n$}
		\STATE Update ${\bm S}_j$ according to Eq.~\eqref{eqn:s2};
		\ENDFOR
		\FOR {$k = 1$ to $t$}
		\STATE Update ${\bm T}_k$ according to Eq.~\eqref{eqn:t2};
		\ENDFOR
		\UNTIL{Convergence}
		\RETURN ${\bm U}$, ${\bm S}$, ${\bm T}$;
	\end{algorithmic}
\end{algorithm}

Due to the nonnegative constraints, we cannot adopt the gradient descent method to optimize the objective function in Eq.~\eqref{eqn:lossctf} any more. Alternatively, we use the multiplicative updating (MU) algorithm \cite{lee2001algorithms} to solve Eq.~\eqref{eqn:lossctf}. To be more specific, MU alternately updates ${\bm U}$, ${\bm S}$ and ${\bm T}$ with the other two being fixed in each iteration. Although the objective function in Eq.~\eqref{eqn:lossctf} is nonconvex over ${\bm U}$, ${\bm S}$ and ${\bm T}$ simultaneously, it is a convex function in each variable when the other two are fixed. Thus we can derive a closed-form update rule for each variable under the Karush-Kuhn-Tucker (KKT) conditions \cite{boyd2004convex}. The detailed update rules are listed as follows: 
\begin{equation} \label{eqn:u2}
{\bm U}_i \leftarrow {\bm U}_i \odot \frac{\sum_{j=1}^n \sum_{k=1}^t \bm{\mathcal{I}}_{ijk} \bm{\Delta}_{ijk} \bm{\mathcal{X}}_{ijk} ({\bm S}_j \odot {\bm T}_k)}{\sum_{j=1}^n \sum_{k=1}^t \bm{\mathcal{I}}_{ijk} \bm{\Delta}_{ijk} \bm{\mathcal{\hat{X}}}_{ijk} ({\bm S}_j \odot {\bm T}_k) + \lambda_u {\bm U}_i},
\end{equation}
\begin{equation} \label{eqn:s2}
{\bm S}_j \leftarrow {\bm S}_j \odot \frac{\sum_{i=1}^m \sum_{k=1}^t \bm{\mathcal{I}}_{ijk} \bm{\Delta}_{ijk} \bm{\mathcal{X}}_{ijk} ({\bm U}_i \odot {\bm T}_k)}{\sum_{i=1}^m \sum_{k=1}^t \bm{\mathcal{I}}_{ijk} \bm{\Delta}_{ijk} \bm{\mathcal{\hat{X}}}_{ijk} ({\bm U}_i \odot {\bm T}_k) + \lambda_s {\bm S}_j},
\end{equation}
\begin{equation} \label{eqn:t2}
{\bm T}_k \leftarrow {\bm T}_k \odot \frac{\sum_{i=1}^m \sum_{j=1}^n \bm{\mathcal{I}}_{ijk} \bm{\Delta}_{ijk} \bm{\mathcal{X}}_{ijk} ({\bm U}_i \odot {\bm S}_j)}{\sum_{i=1}^m \sum_{j=1}^n \bm{\mathcal{I}}_{ijk} \bm{\Delta}_{ijk} \bm{\mathcal{\hat{X}}}_{ijk} ({\bm U}_i \odot {\bm S}_j) + \lambda_t {\bm T}_k}.
\end{equation}
In the above equations, $\bm{\Delta}_{ijk}$ is defined as $\bm{\Delta}_{ijk} = \frac{1}{\gamma^2 + (\bm{\mathcal{X}}_{ijk} - \bm{\mathcal{\hat X}}_{ijk})^2}$.

Algorithm~\ref{alg:bbb} summarizes the overall optimization procedure of our time-aware QoS prediction method. Since  Algorithm~\ref{alg:bbb} updates ${\bm U}$, ${\bm S}$ and ${\bm T}$ alternately and each update decreases the objective function value monotonically,  it is guaranteed to converge to a local minimal solution. The time complexity of Algorithm~\ref{alg:bbb} is shown in Theorem~\ref{the:dynamic}.
\begin{theorem} \label{the:dynamic}
	Let $r'$ denote the number of iterations for Algorithm~\ref{alg:bbb} to achieve convergence and let $\rho'$ denote the number of available entries in $\bm{\mathcal{X}}$, then the time complexity of Algorithm~\ref{alg:bbb} is $\mathcal{O}(r' \rho' l)$.
\end{theorem}
\begin{proof}
	The proof is similar to that of Theorem~\ref{the:static}. In each iteration, it takes $\mathcal{O}(\rho' l)$ time to update $\bm{U}$, $\bm{S}$ and $\bm{T}$. Therefore, the total time complexity of Algorithm~\ref{alg:bbb} is of order $\mathcal{O}(r' \rho' l)$.
\end{proof}


It is worth mentioning that in both Algorithm~\ref{alg:aaa} and Algorithm~\ref{alg:bbb}, we do not detect outliers explicitly. Thus our method will not suffer from the problem of misclassification, which indicates that our method is more resilient and more robust to outliers.


\section{Experiments} \label{sec:exp}

In this section, we conduct a set of experiments on both static QoS prediction and time-aware QoS prediction to evaluate the efficiency and effectiveness of our method by comparing it with several state-of-the-art QoS prediction methods. We implement our method and all baseline methods in Python 3.7. And all the experiments are conducted on a server with two 2.4GHz Intel Xeon CPUs and 128GB main memory running Ubuntu 14.04.5 (64-bit)\footnote{The source code is available at \url{https://github.com/smartyfh/CMF-CTF}}.

\subsection{Datasets}

We conduct all experiments on a publicly available dataset collection--WS-DREAM\footnote{\url{https://github.com/wsdream/wsdream-dataset}}, which was collected from real-world Web services. WS-DREAM contains both static and dynamic QoS datasets. The static dataset  describes real-world QoS measurements, including both response time and throughput  values, obtained from 339 users on 5825 Web services. The dynamic dataset describes real-world QoS measurements from 142 users on 4500 Web services over 64 consecutive time slices (at 15-minute interval). The dynamic dataset also includes records of both response time and throughput values. The statistics of the datasets are presented in Table~\ref{tabdata}.

\begin{table}[t!]
    \small
	\caption{Statistics of QoS Data}
	\vspace*{-0.3cm}
	\begin{center}
        \setlength{\tabcolsep}{0.77mm}
		\begin{tabular}{c|c|c c c c c }
\hline
			\textbf{Type}                     & \textbf{QoS Attributes} & \#\textbf{User} & \#\textbf{Service} & \#\textbf{Time} & \textbf{Range} & \textbf{Mean} \\
\hline
\multirow{2}{*}{Static}  & Response Time (s)  & 339     & 5825       & -      &   0-20    &  0.9086    \\
                         & Throughput (kbps)  & 339     & 5825       & -      &   0-1000    &  47.5617    \\
\hline
\multirow{2}{*}{Dynamic} & Response Time (s)  & 142     & 4500       & 64     &    0-20   &  3.1773   \\
                         & Throughput (kbps)  & 142     & 4500       & 64     &   0-6727    &  11.3449   \\
\hline
		\end{tabular}
		\label{tabdata}
	\end{center}
\end{table}

\subsection{Evaluation Metrics}

The most commonly used evaluation metrics for QoS prediction include mean absolute error (MAE) \cite{zheng2011qos} and root mean square error (RMSE) \cite{zheng2011qos}. Let $\Pi$ denote the set of QoS values to be predicted (i.e., $\Pi$ is the testing set) and let $N = |\Pi|$, then MAE is calculated as:
\begin{equation}
MAE = \frac{\sum_{q \in \Pi} |q - \hat{q}|}{N},
\end{equation}
and RMSE is calculated as:
\begin{equation}
RMSE = \sqrt{\frac{\sum_{q \in \Pi} (q - \hat{q})^2}{N}},
\end{equation}
where $\hat{q}$ denotes the predicted value for the observation $q$. For both MAE and RMSE, smaller values indicate better performance.

However, according to the definition of MAE and RMSE, we can see that both MAE and RMSE are sensitive to outliers, which means that if $\Pi$ contains outliers, then MAE and RMSE cannot truly reflect the QoS prediction performance. For example, suppose that $q_* \in \Pi$ is an outlier, then in order to get a small MAE value and RMSE value, the predicted $\hat{q}_*$ should be close to $q_*$ rather than the normal QoS value. As thus, a smaller MAE or RMSE value may not really indicate better performance. To overcome this limitation, we eliminate outliers from $\Pi$ when calculating MAE and RMSE. Note that we do not have groundtruth labels for outliers. Therefore, we need to detect outliers from scratch. To achieve this goal, we employ the $i$Forest (short for isolation forest) method  \cite{liu2008isolation, liu2012isolation} for outlier detection. $i$Forest detects outliers purely based on the concept of isolation without employing any distance or density measure, making $i$Forest quite efficient and robust. $i$Forest will calculate an outlier score for each datum. The score takes value in the range of $[0, 1]$ and a larger value indicates more possibility to be outliers. Based on the outlier score, we can set the number of outliers flexibly. To intuitively show the effectiveness of $i$Forest, we report the outlier detection results of a randomly selected Web service from the static dataset in Figure~\ref{fig:outlierexp}, where the outlier ratio is set to 0.05. As can be seen, $i$Forest demonstrates good performance in outlier detection. It can detect both overinflated and underinflated outliers.

\begin{figure}[t!]
	\centering
	\subfigure[{Response Time}]{
		\includegraphics[width=0.48\columnwidth]{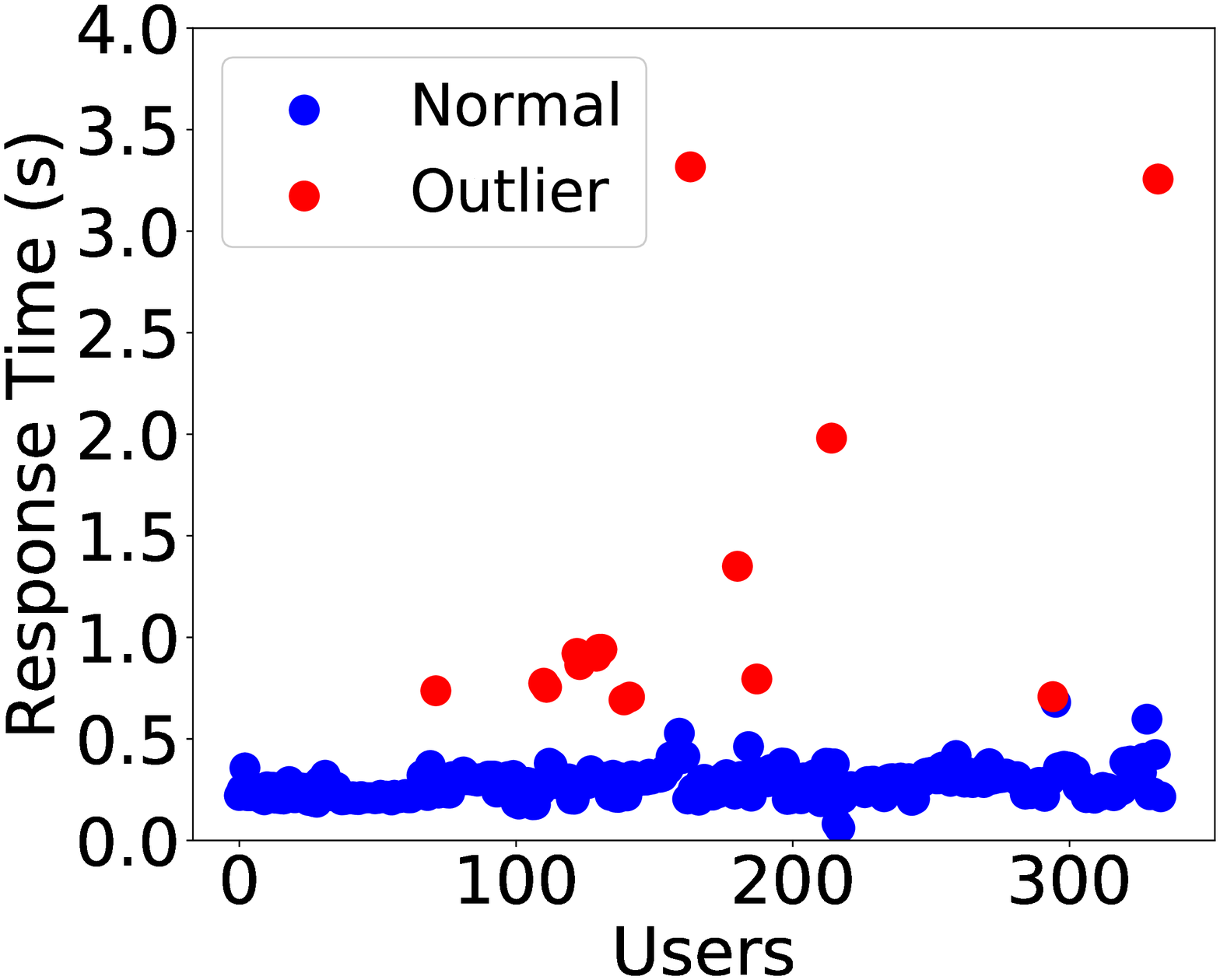}
	}
	\subfigure[{Throughput}]{
		\includegraphics[width=0.472\columnwidth]{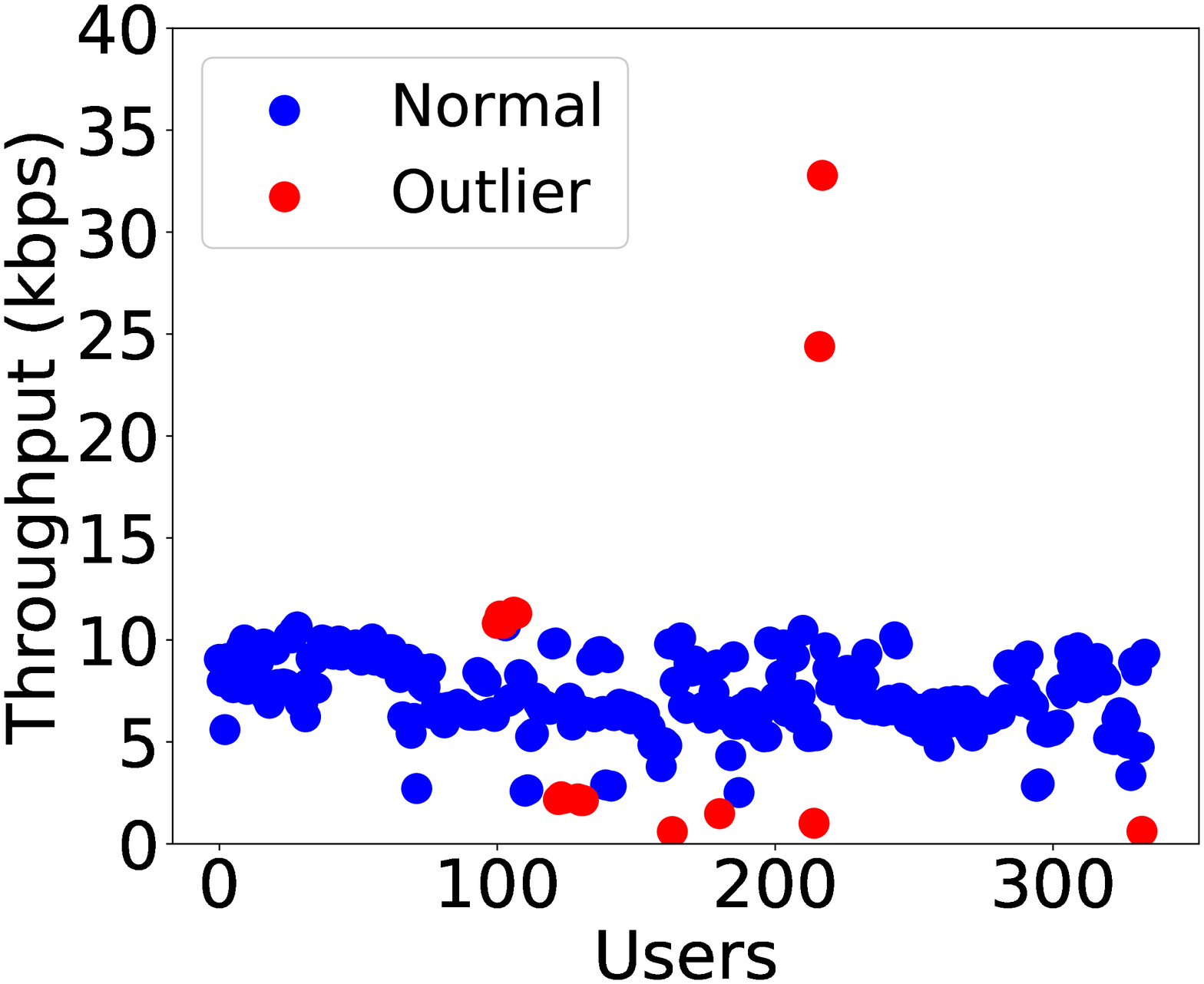}
	}
	\vspace*{-0.3cm}
	\caption{Exemplary outlier detection results by $i$Forest.}
	\label{fig:outlierexp}
	\vspace*{-0.3cm}
\end{figure}

\begin{table*}[t!]
    \small
	\caption{Performance Comparison with Different Training Ratios on Static Dataset (Best Results in Bold Numbers)}
	\vspace*{-0.25cm}
	\begin{center}
        \setlength{\tabcolsep}{1.45mm}
		\begin{tabular}{c|c|c c c c c c|c c c c c c}
			\hline
			\multicolumn{1}{c|}{\multirow{2}*{\textbf{QoS Attributes}}} & \multicolumn{1}{c|}{\multirow{2}*{\textbf{Methods}}} & \multicolumn{6}{c|}{\textbf{MAE}} & \multicolumn{6}{c}{\textbf{RMSE}} \\
			\cline{3-14}
			&  & $10\%$ & $20\%$ & $30\%$ & $70\%$ & $80\%$ & $90\%$ & $10\%$ & $20\%$ & $30\%$ & $70\%$ & $80\%$ & $90\%$ \\
			\hline \hline
			\multirow{5}{*}{Response Time} 
			& MF$_2$ & 0.5334 & 0.4103 & 0.3534 & 0.3044 & 0.2921 & 0.2848 & 0.8407 & 0.6978 & 0.6195 & 0.5742 & 0.5565 & 0.5438 \\
			& MF$_1$ & 0.4041 & 0.4037 & 0.4036 & 0.2815 & 0.2786 & 0.2798 & 0.6120 & 0.6106 & 0.6103 & 0.5590 & 0.5544 & 0.5505 \\
			& CAP & 0.3603 & 0.3521 & 0.3312 & 0.2282 & 0.2112 & 0.1881 & 0.6439 & 0.6640 & 0.6789 & 0.5815 & 0.5987 & 0.5821 \\
			& TAP & 0.3385 & 0.2843 & 0.2449 & 0.2477 & 0.2812 & 0.3189 & 0.5512 & 0.4985 & 0.4589 & 0.4687 & 0.5155 & 0.5665 \\
            & DALF & 0.3955 & 0.3439 & 0.3081 & 0.2496 & 0.2492 & 0.2397 & 0.7466 & 0.6779 & 0.5974 & 0.5471 & 0.5403 & 0.5388 \\
			& \textbf{CMF} & \textbf{0.1762} & \textbf{0.1524} & \textbf{0.1408} & \textbf{0.1153} & \textbf{0.1102} & \textbf{0.1085} & \textbf{0.3705} & \textbf{0.3599} & \textbf{0.3504} & \textbf{0.3106} & \textbf{0.2877} & \textbf{0.2699} \\
			\hline \hline
			\multirow{5}{*}{Throughput} 
			& MF$_2$ & 13.9730 & 12.3750 & 10.7753 & 7.8371 & 7.8255 & 7.8071 & 28.9608 & 26.8906 & 24.6608 & 19.6406 & 19.2831 & 18.6451 \\
			& MF$_1$ & 16.5509 & 13.1105 & 10.7200 & 7.5736 & 7.3263 & 7.1115 & 33.8889 & 27.9648 & 23.6611 & 18.1316	& 17.9458 & 17.3698 \\
			& CAP & 16.4269 & 16.3125 & 16.1946 & 9.7147 & 8.6984 & 7.8516 & 32.9558 & 32.9334 & 32.9540 & 23.7955 & 22.2425 & 21.3711 \\
			& TAP & 22.1419 & 19.8273 & 17.8388 & 14.5786 & 14.8380 & 15.4028 & 43.4987 & 40.9533 & 38.8371 & 33.3052 & 32.4076 & 32.0935 \\
            & DALF & 13.1968 & 11.9619 & 10.6882 & 7.8156 & 7.7902 & 7.7771 & 27.8531 & 26.0299 & 24.4506 & 19.3523 & 18.9886 & 18.2965 \\
			& \textbf{CMF} & \textbf{8.4573} & \textbf{7.2501} & \textbf{6.4300} & \textbf{5.1865} & \textbf{5.1241} & \textbf{5.0078} & \textbf{24.9137} & \textbf{20.8927} & \textbf{18.8985} & \textbf{17.2916}  & \textbf{17.1433}  & \textbf{16.9388} \\
			\hline
		\end{tabular}
		\label{tab1}
	\end{center}
\end{table*}

\subsection{Baseline Methods}

For ease of presentation, we name our method for static QoS prediction as \textbf{CMF} and our method for time-aware QoS prediction as \textbf{CTF} hereafter. For static QoS prediction, we compare CMF with the following five methods:
\begin{itemize}
	
	\item \textbf{MF$_2$}: MF$_2$ denotes the basic MF-based QoS prediction method \cite{zheng2013personalized} and it measures the discrepancy between the observed QoS values and the predicted ones by $L_2$-norm.
	
	\item \textbf{MF$_1$}: MF$_1$ is also an MF-based QoS prediction method \cite{zhu2018similarity}. However, it utilizes the $L_1$-norm loss to construct the objective function. MF$_1$ is expected to be more robust to outliers. Note that we implement MF$_1$ a little differently from the original one proposed in \cite{zhu2018similarity}. In our implementation, we ignore the privacy and location information. 
	
	\item \textbf{CAP}: CAP is a credibility-aware QoS prediction method \cite{wu2015qos}. It first employs a two-phase $k$-means clustering algorithm to identify untrustworthy users (i.e., outliers), and then predicts unknown QoS values based on the clustering information contributed by trustworthy users.
	
	\item \textbf{TAP}: TAP is a trust-aware QoS prediction method \cite{su2017tap}. It aims to provide reliable QoS prediction results via calculating the reputation of users by a beta reputation system, and it identifies outliers based on $k$-means clustering as well.

   \item \textbf{DALF}: DALF is a data-aware latent factor model for QoS prediction \cite{wu2019data}. It utilizes the density peaks based clustering algorithm \cite{rodriguez2014clustering} to detect unreliable QoS data directly.
	
\end{itemize}

For time-aware QoS prediction, we compare our CTF with the following five methods:
\begin{itemize}
	
	\item \textbf{NNCP}: NNCP is a tensor-based time-aware QoS prediction method \cite{zhang2014temporal}. It is based on CP decomposition and imposes nonnegative constraints on all the factor matrices.
	
	\item \textbf{BNLFT}: BNLFT is a biased nonnegative tensor factorization model \cite{luo2019temporal}. It incorporates linear biases into the model for describing QoS fluctuations, and it adds nonnegative constraints to the factor matrices as well.
	
	\item \textbf{WLRTF}: WLRTF is an MLE (maximum likelihood estimation) based tensor factorization method \cite{chen2016robust}. It models the noise of each datum  as a mixture of Gaussian (MoG). 

    \item \textbf{PLMF}: PLMF is an LSTM (long short-term memory) \cite{hochreiter1997long} based  QoS prediction method \cite{xiong2018personalized}. PLMF can capture the dynamic latent representations of users
and Web services.

    \item \textbf{TASR}: TASR is a time-aware QoS  prediction method \cite{ding2018time}. It integrates similarity-enhanced collaborative filtering model and the ARIMA model (a time series analysis model)  \cite{box1970distribution}. 
\end{itemize}

\begin{table*}[t!]
\small
	\caption{Performance Comparison with Different Outlier Ratios on Static Dataset (Best Results in Bold Numbers)}
	\vspace*{-0.25cm}
	\begin{center}
        \setlength{\tabcolsep}{1.45mm}
		\begin{tabular}{c|c|c c c c c c|c c c c c c}
			\hline
			\multicolumn{1}{c|}{\multirow{2}*{\textbf{QoS Attributes}}} & \multicolumn{1}{c|}{\multirow{2}*{\textbf{Methods}}} & \multicolumn{6}{c|}{\textbf{MAE}} & \multicolumn{6}{c}{\textbf{RMSE}} \\
			\cline{3-14}
			&  & $2\%$ & $4\%$ & $6\%$ & $8\%$ & $10\%$ & $20\%$ & $2\%$ & $4\%$ & $6\%$ & $8\%$ & $10\%$ & $20\%$ \\
			\hline \hline
			\multirow{5}{*}{Response Time}
			& MF$_2$ & 0.4080 & 0.3732 & 0.3533 & 0.3445 & 0.3306 & 0.3072 & 0.8040 & 0.7210 & 0.6711 & 0.6508 & 0.6021 & 0.5810 \\
			& MF$_1$ & 0.3761 & 0.3390 & 0.3185 & 0.2972 & 0.2702 & 0.2525 & 0.7935 & 0.6903 & 0.6398 & 0.5918 & 0.5575 & 0.3744 \\
			& CAP & 0.4163 & 0.3657 & 0.3311 & 0.2997 & 0.2739 & 0.2413 & 0.9789 & 0.8375 & 0.7616 & 0.6817 & 0.6191 & 0.5258 \\
			& TAP & 0.4562 & 0.3788 & 0.3268 & 0.2703 & 0.2183 & 0.1649 & 1.1536 & 0.9393 & 0.8148 & 0.6475 & 0.4294 & 0.2478 \\
            & DALF & 0.3622 & 0.3217 & 0.3071 & 0.2890 & 0.2781 & 0.2480 & 0.7695 & 0.6728 & 0.6346 & 0.5975 & 0.5701 & 0.5132 \\
			& \textbf{CMF} & \textbf{0.2134} & \textbf{0.1758} & \textbf{0.1545} & \textbf{0.1384} & \textbf{0.1253} & \textbf{0.1019} & \textbf{0.6582} & \textbf{0.5001} & \textbf{0.4452} & \textbf{0.3811} & \textbf{0.3195} & \textbf{0.2347} \\
			\hline \hline
			\multirow{5}{*}{Throughput} 
			& MF$_2$ & 11.8832 & 10.7024 & 9.6776 & 9.0889 & 8.6373 & 8.1358 & 32.9795 & 28.5992 & 25.3608 & 22.9710 & 21.1042 & 18.5597 \\
			& MF$_1$ & 12.3647 & 10.7403 & 9.8674 & 9.2223 & 8.7708 & 8.1667 & 32.9672 & 27.7982 & 24.4438 & 22.1691 & 20.2018 & 17.4015 \\
			& CAP & 18.2991 & 16.8273 & 15.5975 & 13.8889 & 13.6477 & 12.6762 & 45.9353 & 39.6390 & 35.5944 & 31.4784 & 29.2029 & 25.2830 \\
			& TAP & 22.0584 & 18.8479 & 16.9577 & 15.9026 & 15.1283 & 14.2151 & 58.5192 & 47.7490 & 41.6689 & 38.5700 & 35.2813 & 31.2779 \\
            & DALF & 11.8763 & 10.5724 & 9.1783 & 8.9276 & 8.6037 & 8.0449 & 32.8586 & 28.5797 & 24.8752 & 22.7428 & 20.9789 & 18.3713 \\
			& \textbf{CMF} & \textbf{8.3266} & \textbf{7.2138} & \textbf{6.5143} & \textbf{6.0463} & \textbf{5.5718} & \textbf{5.0177} & \textbf{30.5885} & \textbf{26.0933} & \textbf{22.9529} & \textbf{20.7105} & \textbf{17.8538} & \textbf{14.7925} \\
			\hline
		\end{tabular}
		\label{tab2}
	\end{center}
\end{table*}

Although MF$_1$, CAP, TAP and DALF are able to deal with outliers to some extent for static QoS prediction, to our best knowledge, our method CTF is the first to take outliers into consideration for time-aware QoS prediction. It is also worth emphasizing that our method and all baseline methods (except CAP, TAP and DALF) will not explicitly detect outliers when learning the prediction model. The reason for detecting outliers during the testing phase is to make MAE and RMSE be able to truly reflect the QoS prediction performance. For all methods, outliers will be removed when calculating MAE and RMSE. In addition, in the experiments, we run each method 10 times and report the average results for fair comparison.

\subsection{Experiments for Static QoS Prediction}

\subsubsection{Parameter Settings}

In the experiments, for all baseline methods, we tune the corresponding parameters following the guidance of the original papers. As for our method CMF, on the response time dataset, the parameters are set as $l=30$, $\gamma = 1$, $\lambda_u = \lambda_s = 1$, and $\eta_u = \eta_s = 0.003$. On the throughput dataset, the parameters are set as $l=30$, $\gamma = 20$, $\lambda_u = \lambda_s = 0.01$, and $\eta_u = \eta_s = 0.025$. For MF$_2$, MF$_1$ and DALF, the feature dimensionality is also set to 30.

\subsubsection{Experimental Results}

We first report the results by varying the training ratios in the range of $\{0.1, 0.2, 0.3, 0.7, 0.8, 0.9\}$. This is to simulate various prediction scenarios with different data sparsity. For example, when the training ratio is set to 0.1, then $10\%$ of the dataset will be used as training data and the rest will be used as testing data. As aforementioned, during the testing phase, outliers should be eliminated explicitly. Here we set the outlier ratio to 0.1, which means $10\%$ of the testing data with large outlier scores will be removed when calculating MAE and RMSE. The detailed comparison results are presented in Table~\ref{tab1}. As can be seen, our method CMF consistently shows better performance than all baseline methods. Moreover, the MAE and RMSE values obtained by our method are much smaller than those of baseline methods, especially on the response time dataset. For instance, CMF achieves more than $30\%$ performance promotion on response time over both MAE and RMSE. From Table~\ref{tab1}, we can also see that MF$_1$, MF$_2$, DALF and CMF tend to obtain smaller MAE and RMSE values as the training ratio increases. This is desired because a larger training ratio indicates that more QoS observations (i.e., more information) will be used to train the prediction model. However, CAP and TAP do not show this pattern, especially on the response time dataset. We can also observe that although CAP, TAP and DALF explicitly take outliers into consideration during the training phase, their performance is not satisfactory. The reason may be the misclassification of outliers. Since our method does not detect outliers directly during the training phase, it will not suffer from the misclassification issue. The resilience of our method to outliers makes it more robust.

\begin{figure*}[t!]
	\centering
	\subfigure[{MAE}]{
		\includegraphics[width=0.493\columnwidth]{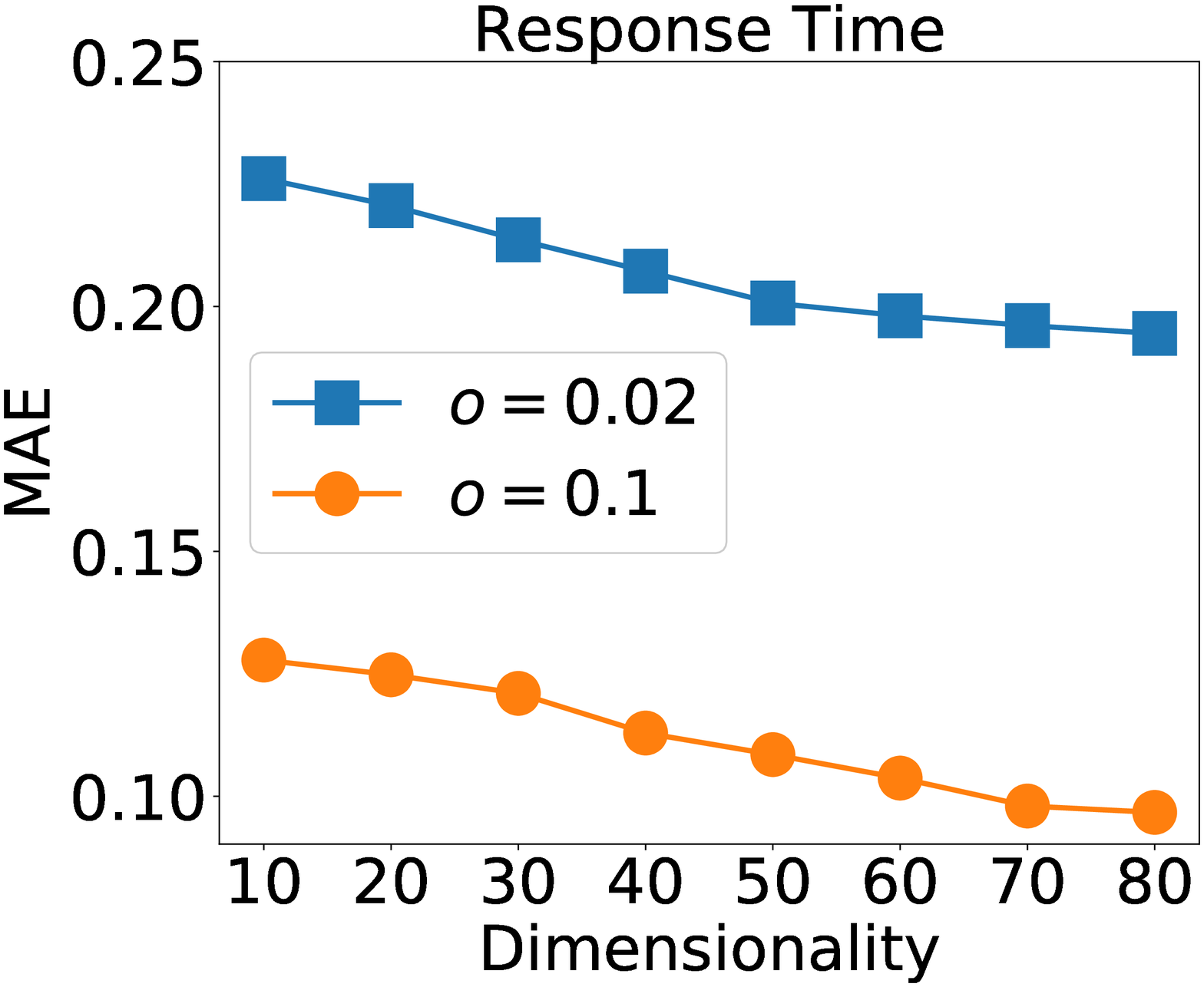}
	}
    \subfigure[{RMSE}]{
		\includegraphics[width=0.478\columnwidth]{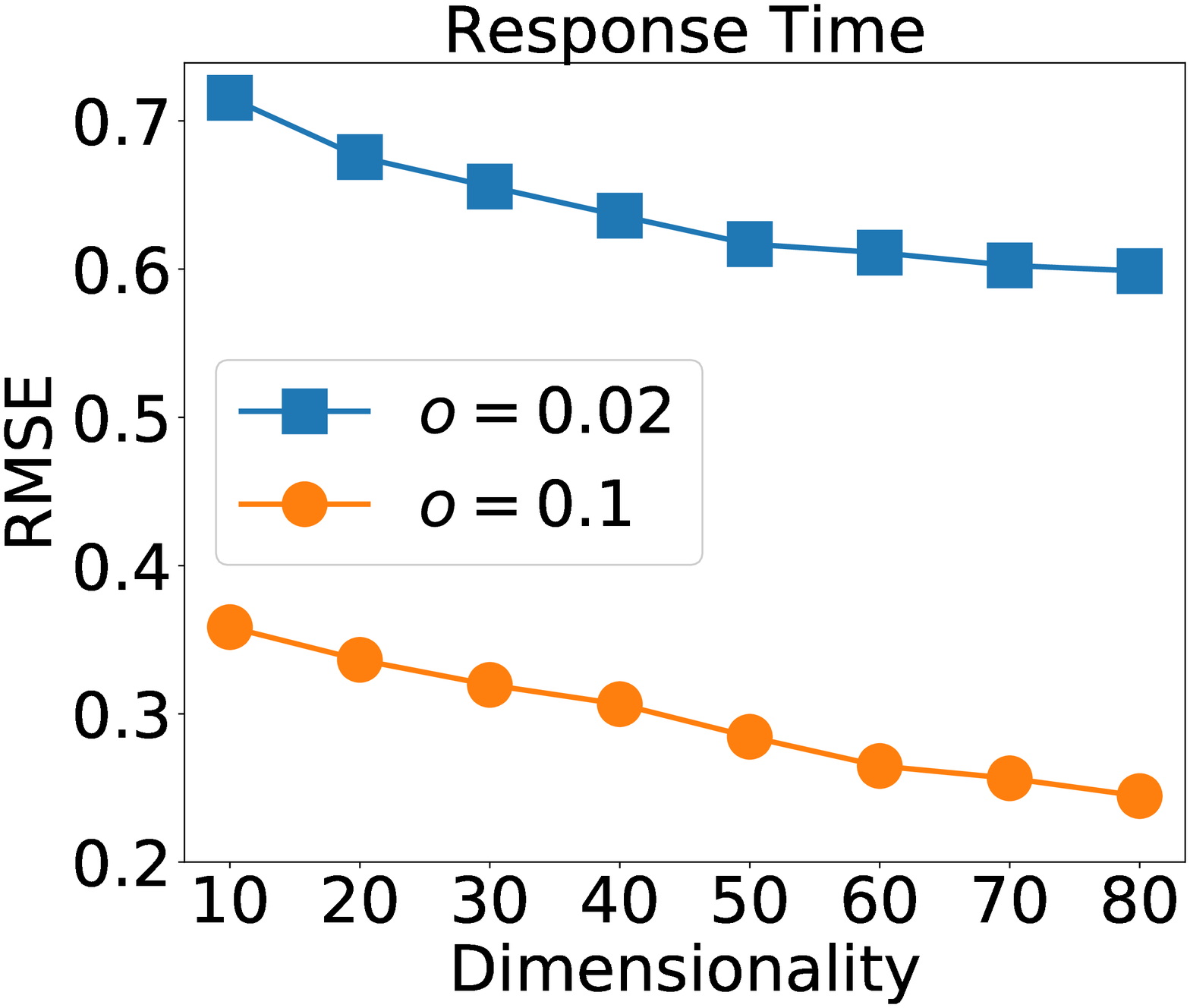}
	}
	\subfigure[{MAE}]{
		\includegraphics[width=0.454\columnwidth]{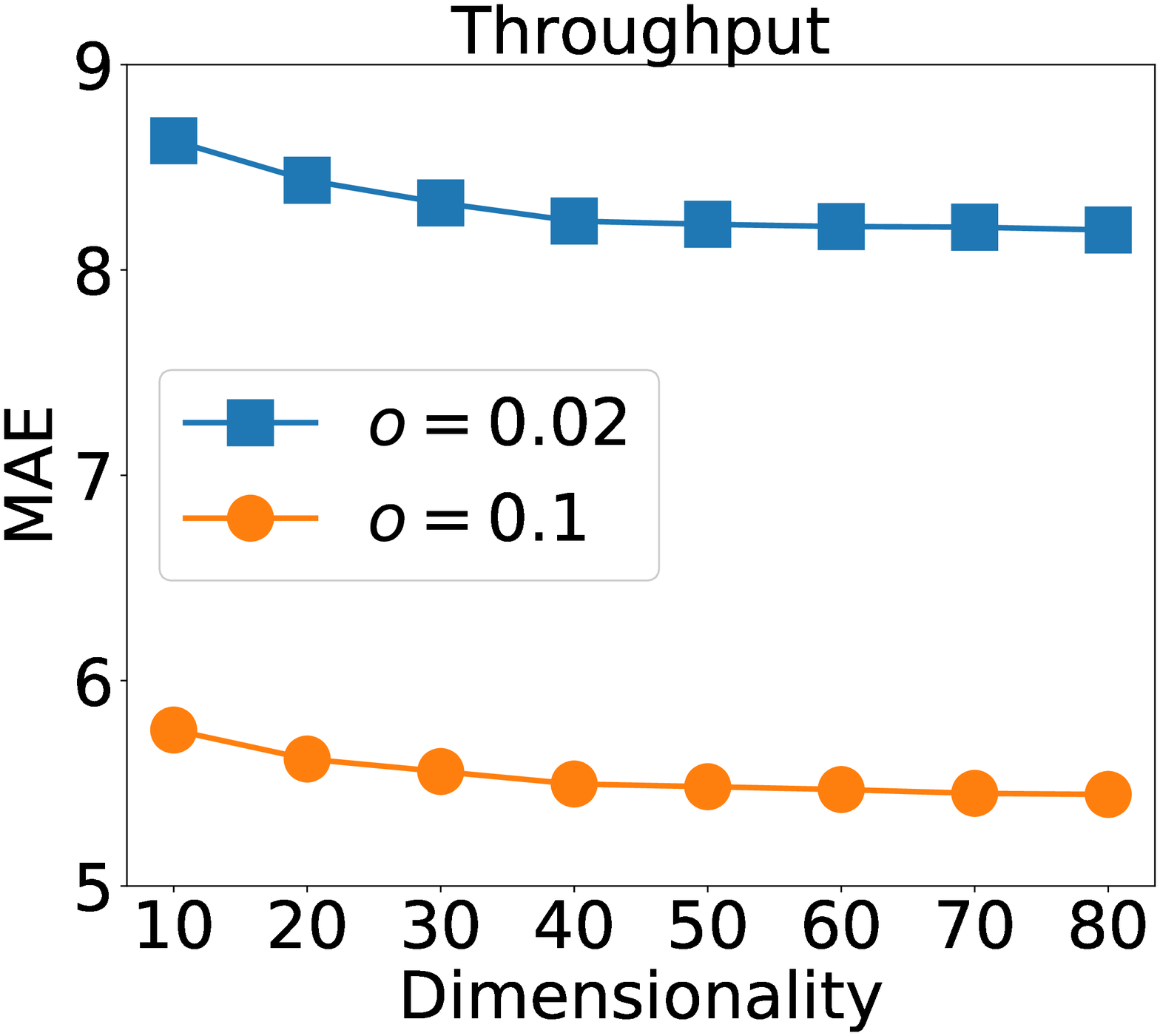}
	}
    \subfigure[{RMSE}]{
		\includegraphics[width=0.47\columnwidth]{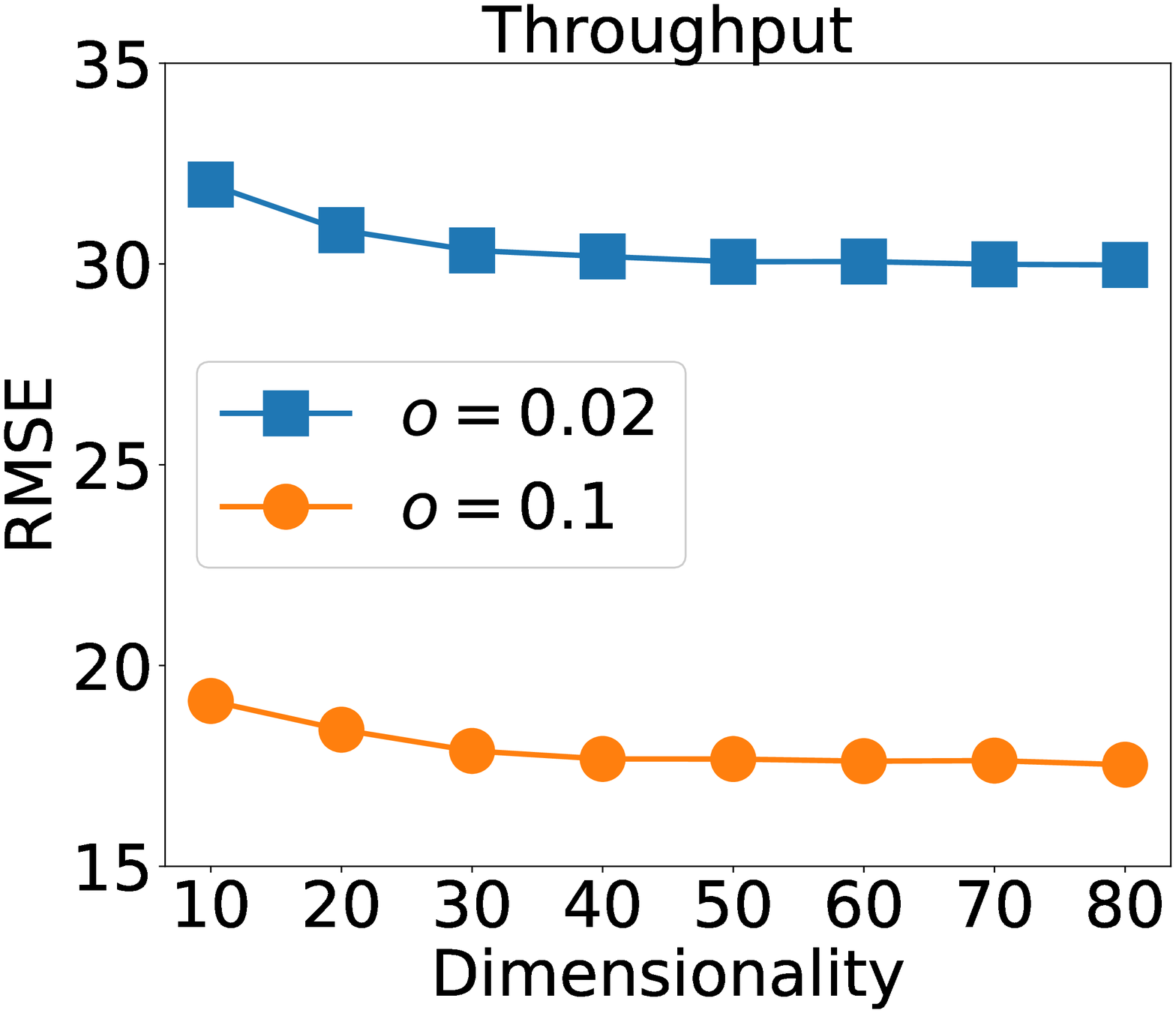}
	}
	\vspace*{-0.35cm}
	\caption{Impact of dimensionality $l$ on CMF (with outlier ratio set to 0.02 and 0.1).}
	\label{fig:dl}
\end{figure*}

\begin{figure*}[t!]
	\centering
	\subfigure[{MAE}]{
		\includegraphics[width=0.478\columnwidth]{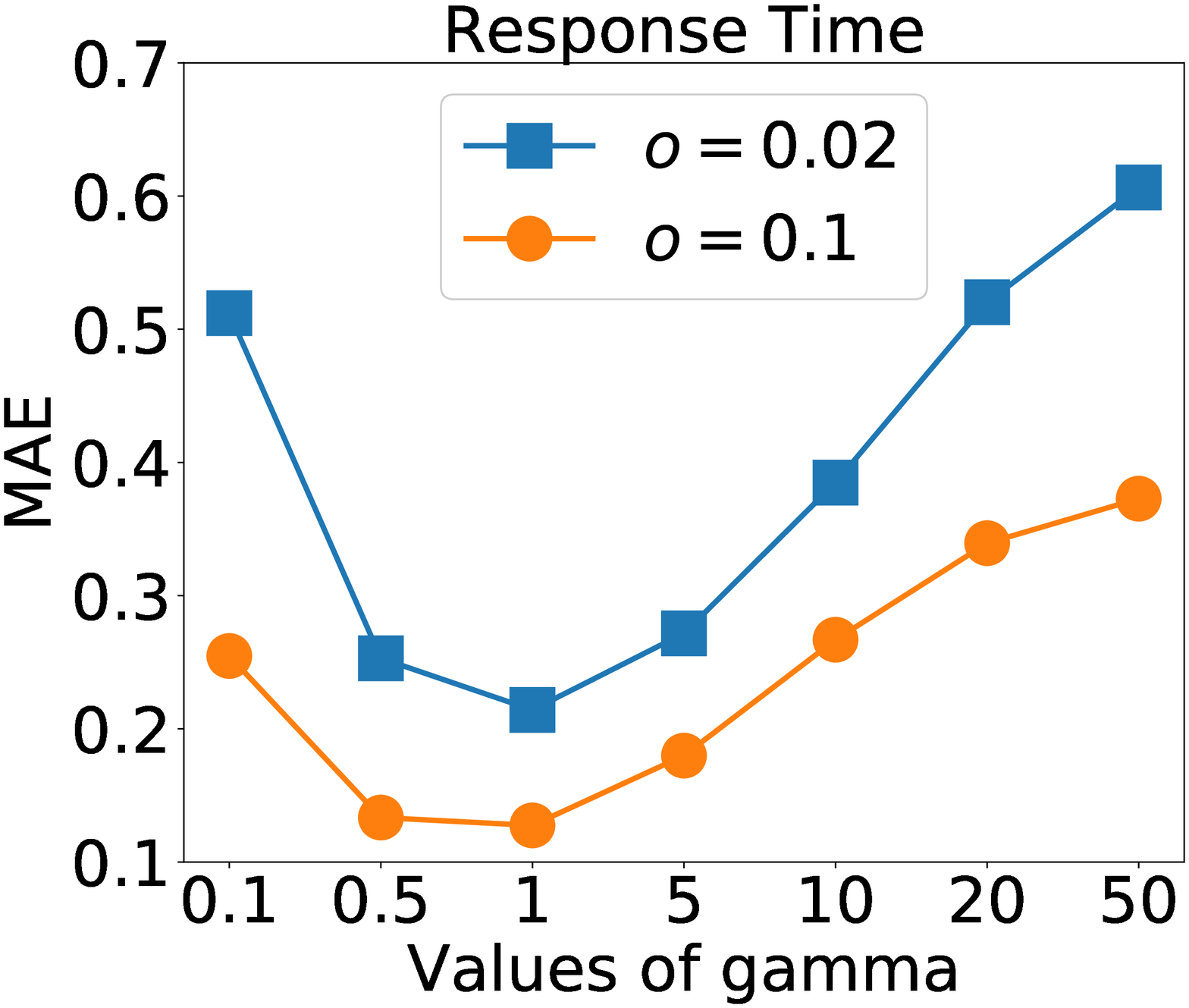}
	}
    \subfigure[{RMSE}]{
		\includegraphics[width=0.478\columnwidth]{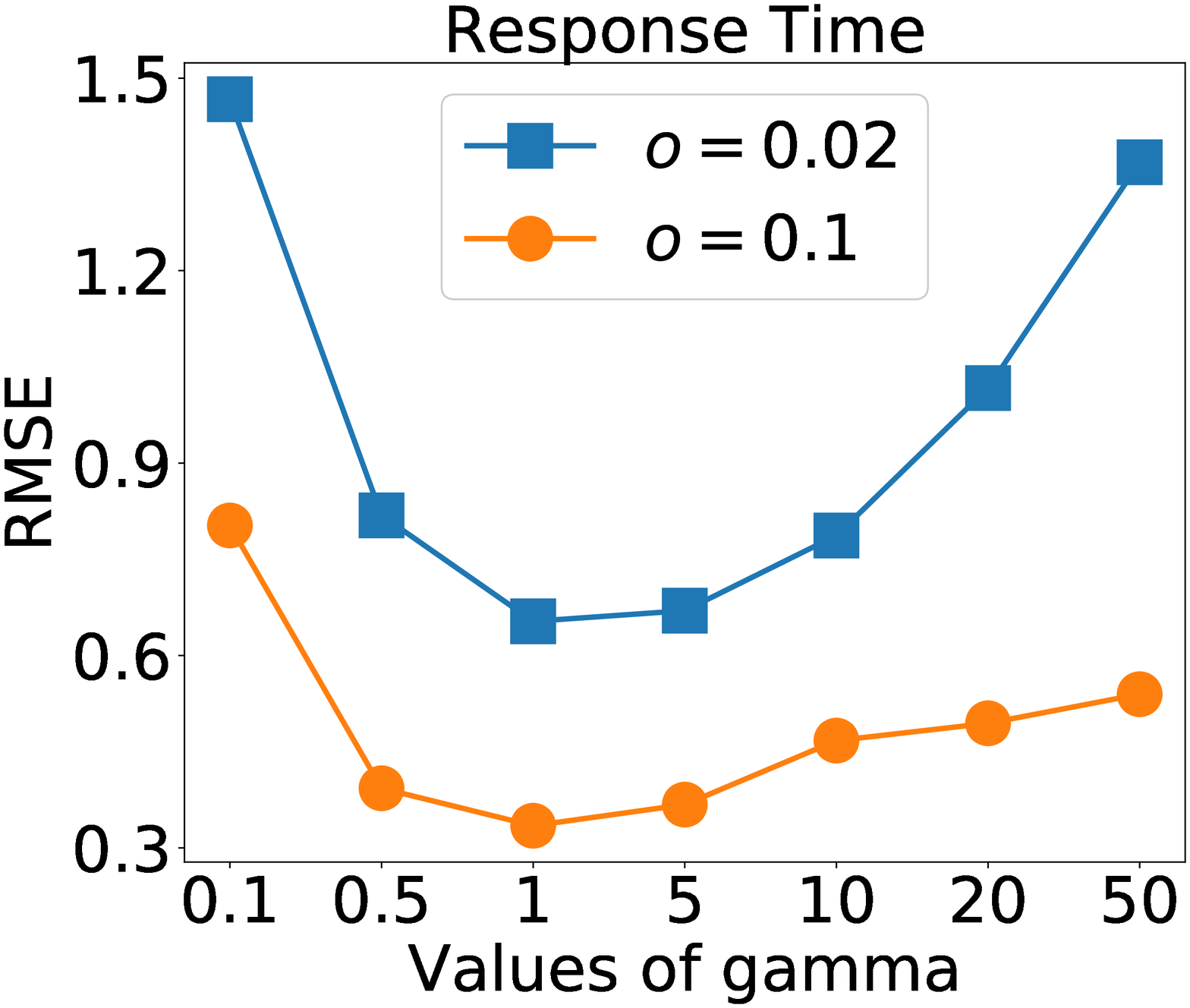}
	}
	\subfigure[{MAE}]{
		\includegraphics[width=0.47\columnwidth]{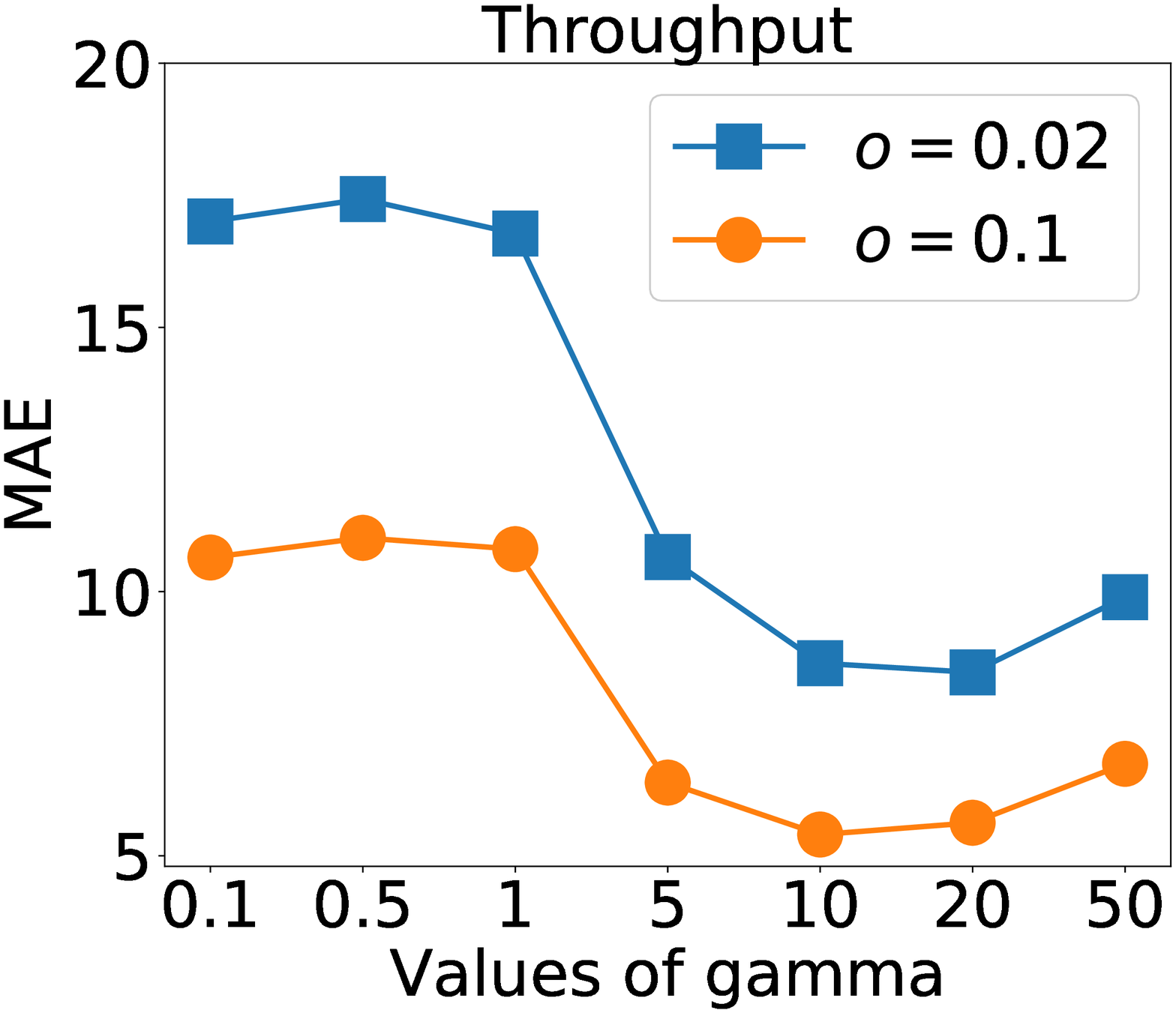}
	}
    \subfigure[{RMSE}]{
		\includegraphics[width=0.47\columnwidth]{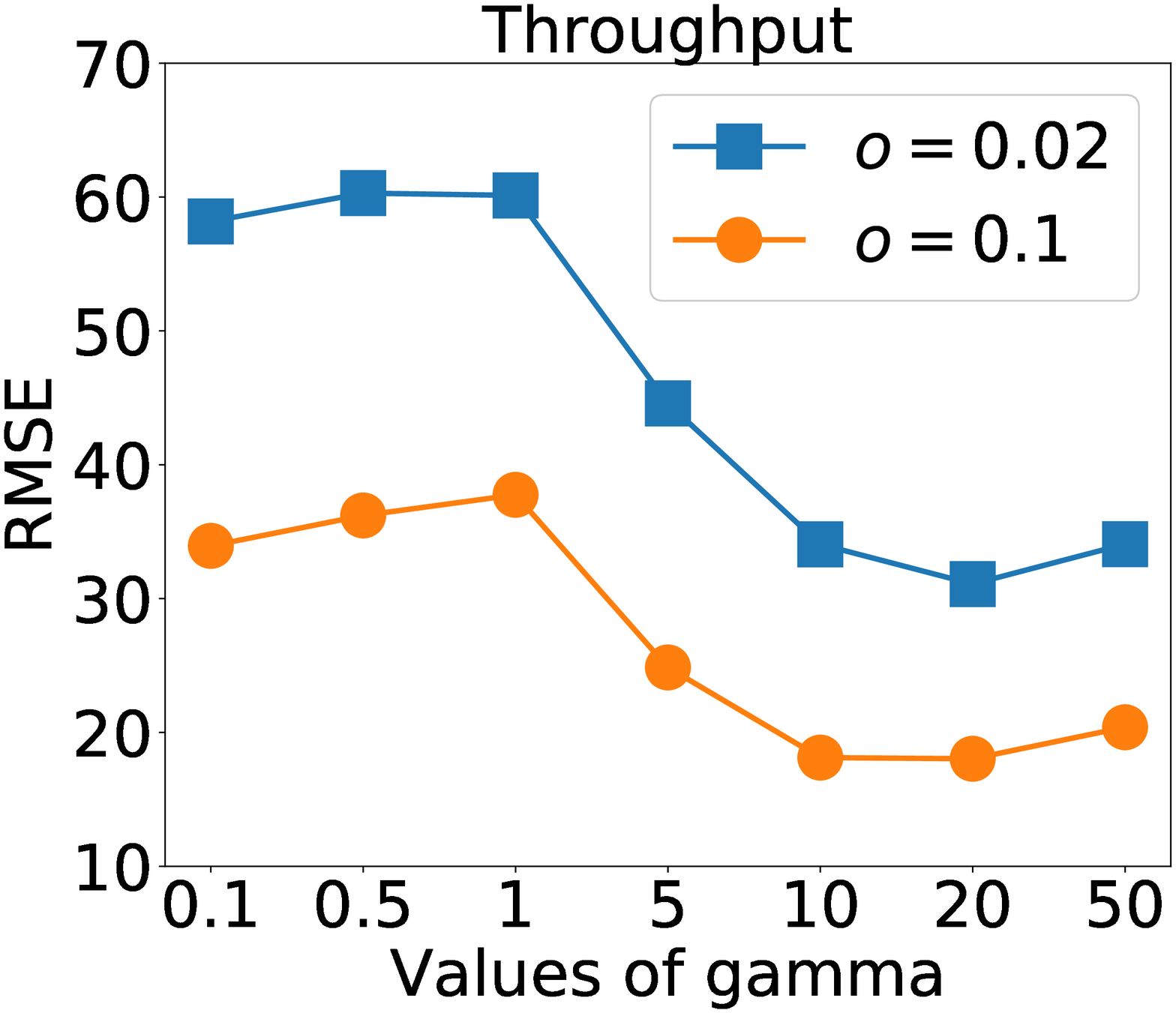}
	}
	\vspace*{-0.35cm}
	\caption{Impact of parameter $\gamma$ on CMF (with outlier ratio set to 0.02 and 0.1).}
	\label{fig:gamma}
\end{figure*}

\begin{figure*}[t!]
	\centering
	\subfigure[{MAE}]{
		\includegraphics[width=0.493\columnwidth]{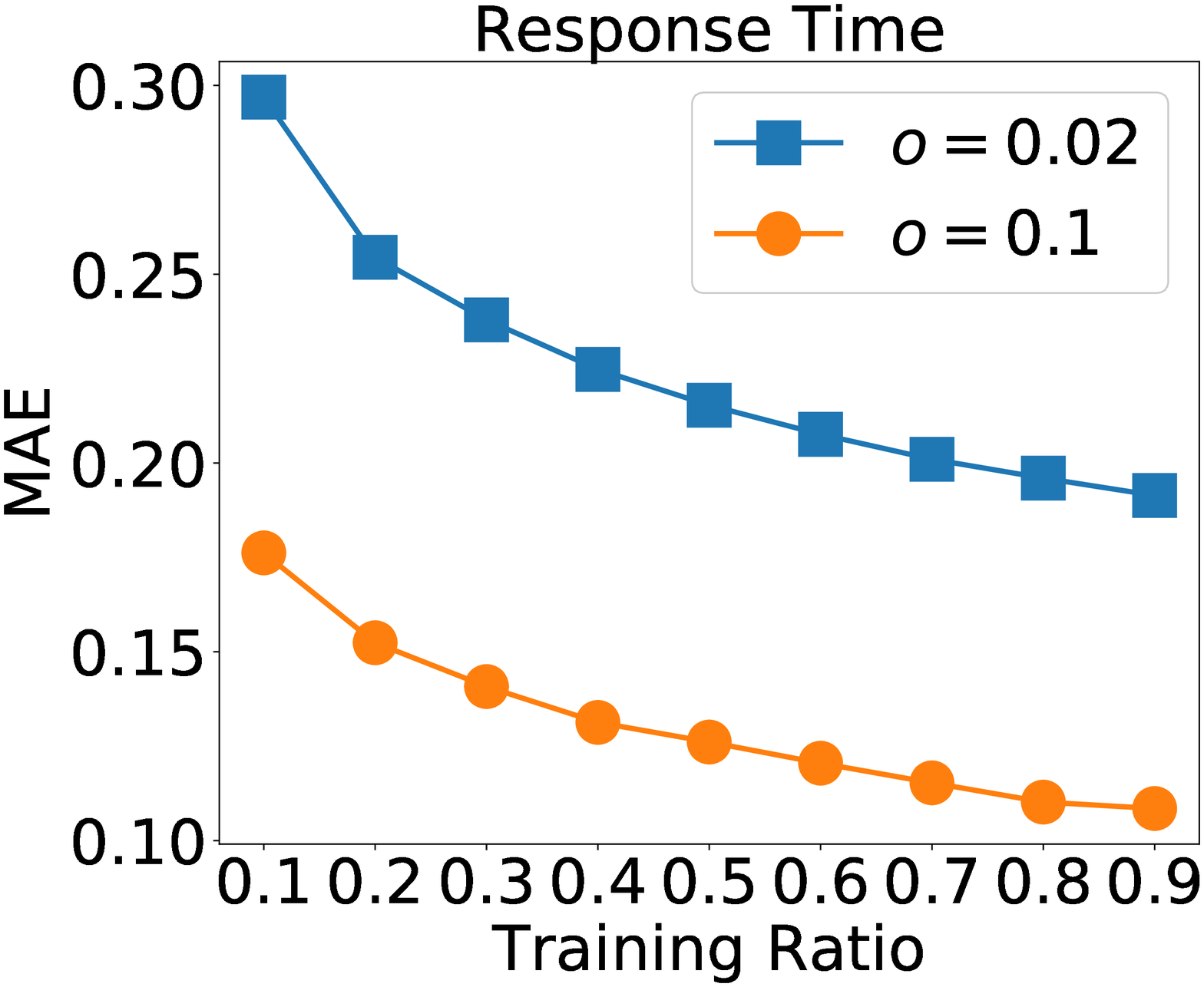}
	}
    \subfigure[{RMSE}]{
		\includegraphics[width=0.478\columnwidth]{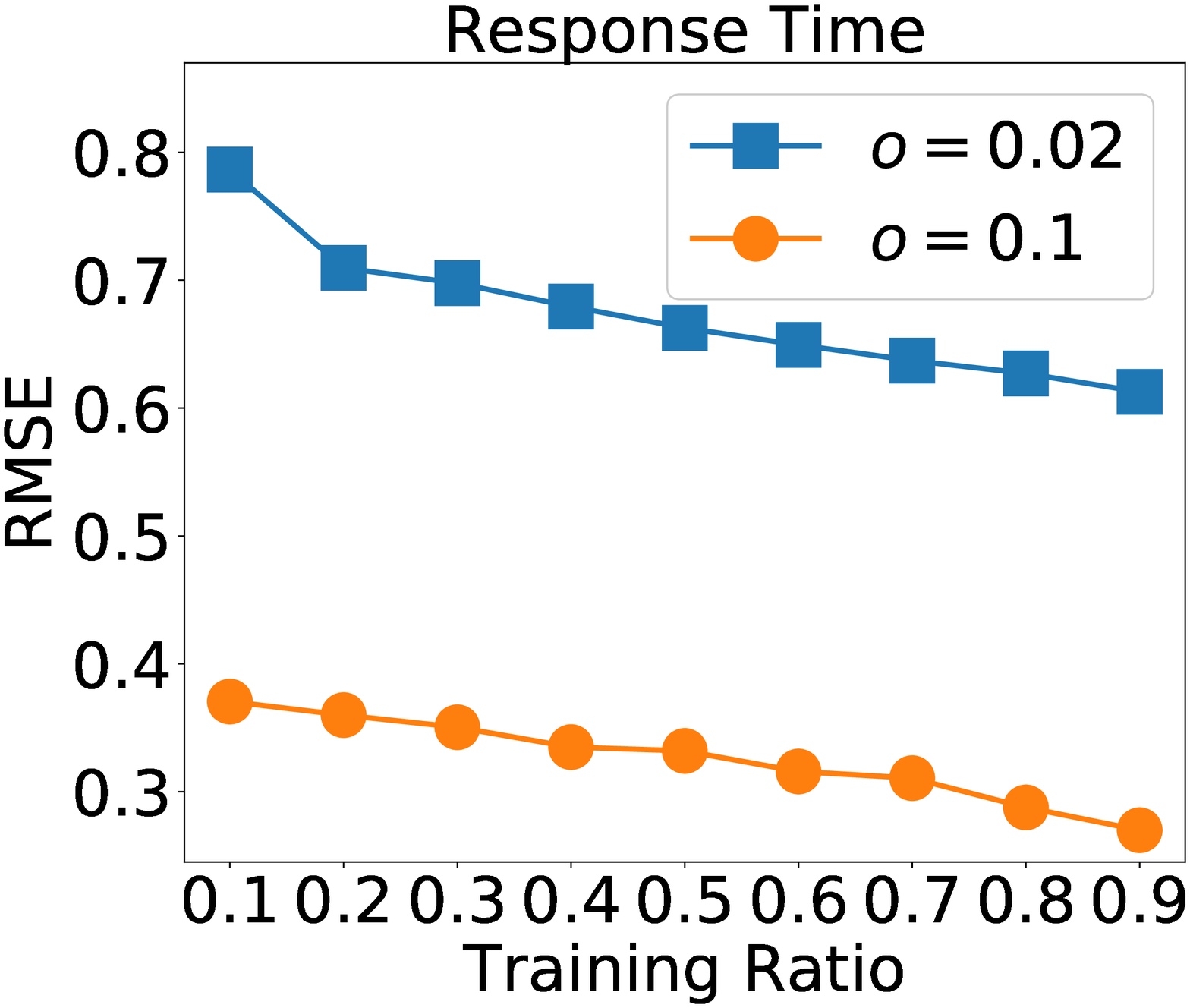}
	}
    \subfigure[{MAE}]{
		\includegraphics[width=0.4691\columnwidth]{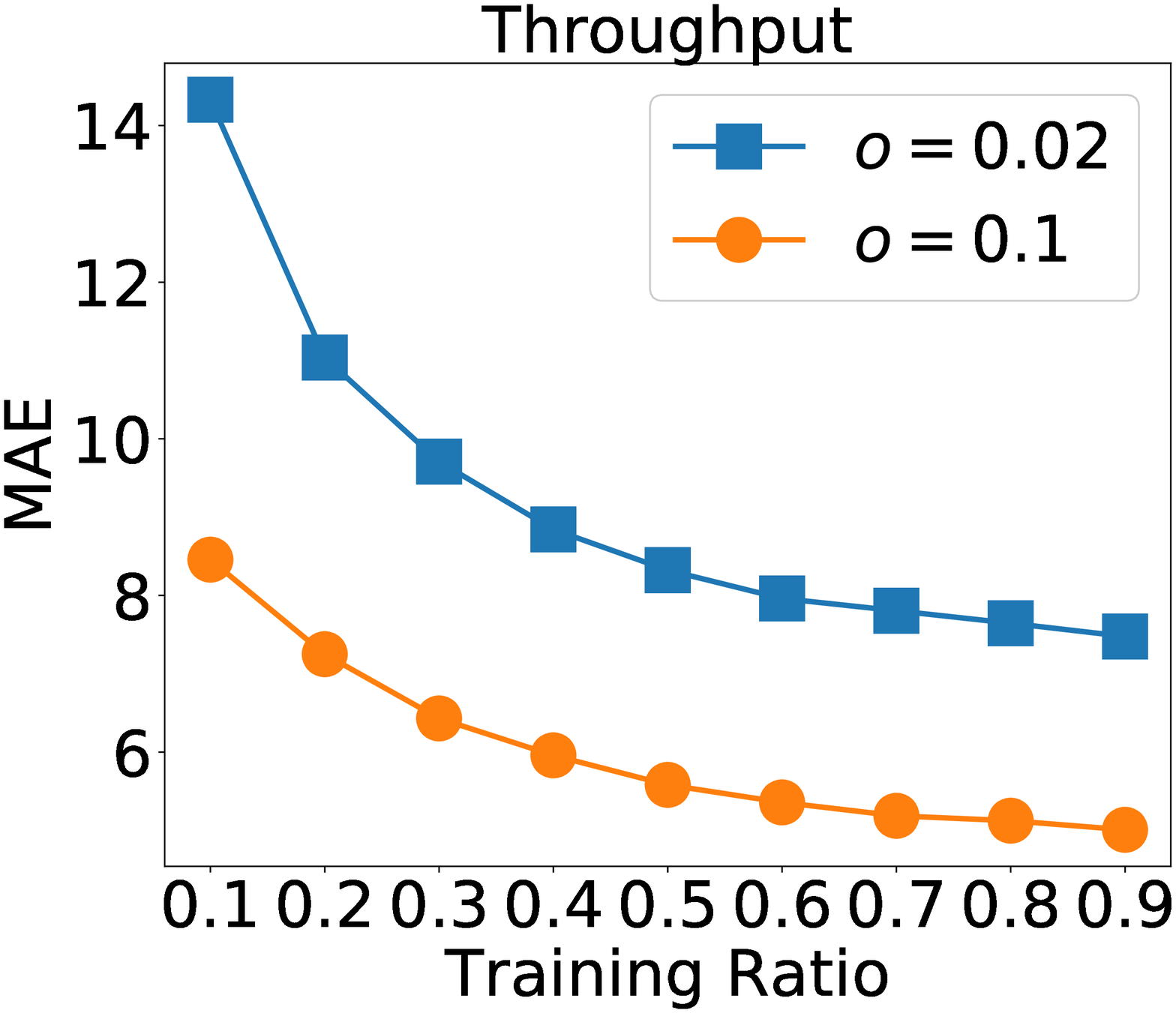}
	}
   \subfigure[{RMSE}]{
		\includegraphics[width=0.4691\columnwidth]{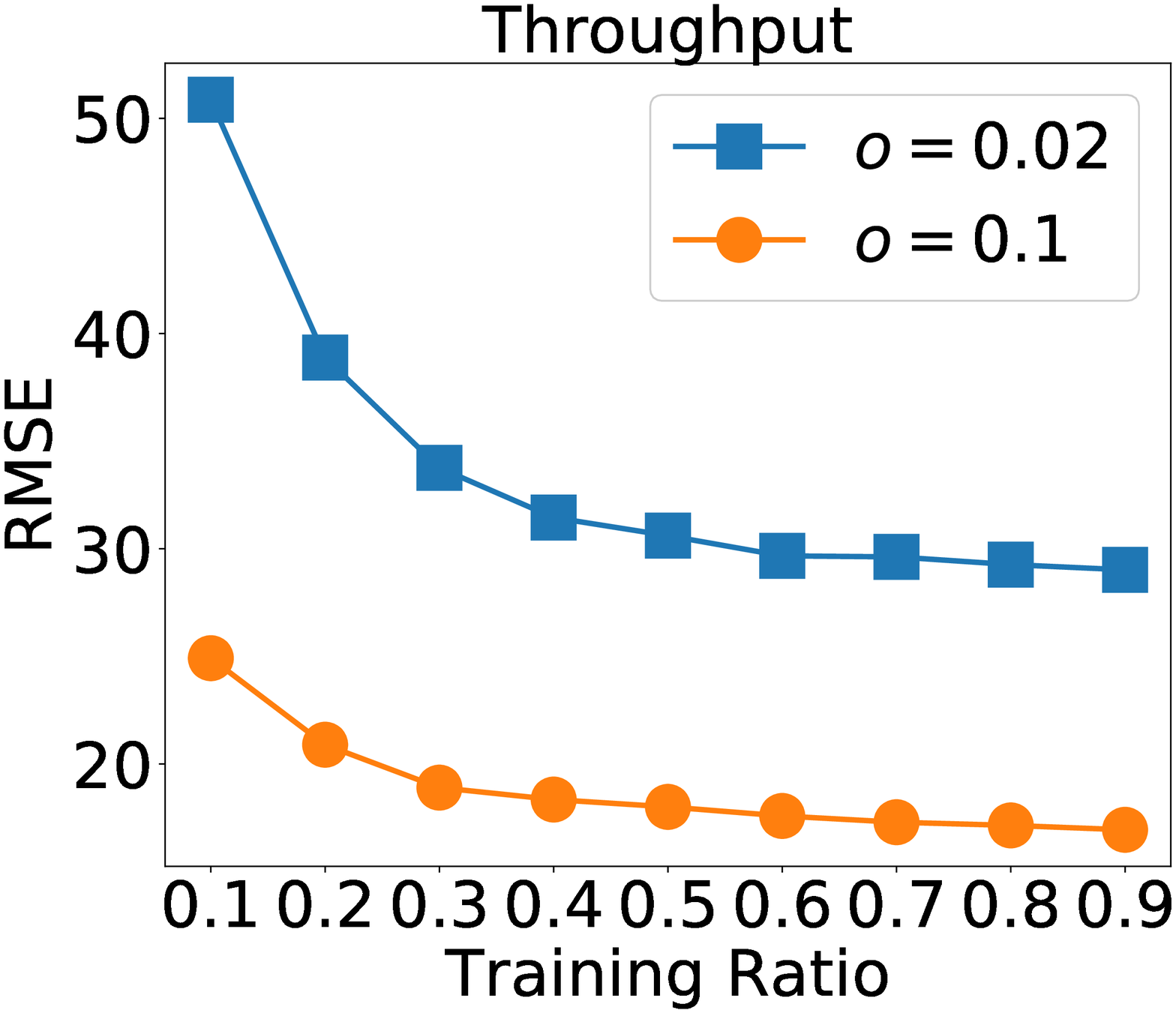}
	}
	\vspace*{-0.1cm}
	\caption{Impact of data sparsity on CMF (with outlier ratio set to 0.02 and 0.1).}
	\label{fig:ratio}
\end{figure*}

We then report the results by varying the outlier ratios in the range of $\{0.02, 0.04, 0.06, 0.08, 0.1, 0.2\}$. In this experiment, the training ratio is fixed at 0.5. The results are shown in Table~\ref{tab2}. From Table~\ref{tab2}, we can see that our method still shows the best performance under different outlier ratios. It can also be observed that the MAE and RMSE values of all methods become smaller as the outlier ratio increases. This is reasonable because 
the larger the outlier ratio is, the more testing data with large outlier scores will be removed. Thus, the effects of outliers on the calculation of MAE and RMSE will be reduced accordingly. From Table~\ref{tab2}, we can further obtain that with the increasing of outlier ratios, the performance promotion of CMF relative to MF$_2$ increases from $48\%$ to $67\%$ over MAE and from $18\%$ to $60\%$ over RMSE on the response time dataset. On the throughput dataset, the performance promotion also increases from $30\%$ to $38\%$ over MAE and from $7\%$ to $20\%$ over RMSE. The increase of performance promotion verifies the necessity of removing outliers during the testing phase. It also verifies the robustness of our proposed method again.

\begin{table*}[t!]
    \small
	\caption{Performance Comparison with Different Training Ratios on Dynamic Dataset (Best Results in Bold Numbers)}
	\vspace*{-0.3cm}
	\begin{center}
        \setlength{\tabcolsep}{1.55mm}
		\begin{tabular}{c|c|c c c c c c|c c c c c c}
			\hline
			\multicolumn{1}{c|}{\multirow{2}*{\textbf{QoS Attributes}}} & \multicolumn{1}{c|}{\multirow{2}*{\textbf{Methods}}} & \multicolumn{6}{c|}{\textbf{MAE}} & \multicolumn{6}{c}{\textbf{RMSE}} \\
			\cline{3-14}
			&  & $10\%$ & $20\%$ & $30\%$ & $70\%$ & $80\%$ & $90\%$ & $10\%$ & $20\%$ & $30\%$ & $70\%$ & $80\%$ & $90\%$ \\
			\hline \hline
			\multirow{5}{*}{Response Time} 
			& NNCP & 1.0796 & 1.0536 & 1.0550 & 1.0424 & 1.0406 & 1.0392 & 2.6401 & 2.5797 & 2.5809 & 2.5574 & 2.5668 & 2.5616 \\
			& BNLFT & 1.0828 & 1.0575 & 1.0467 & 1.0368 & 1.0556 & 1.0403 & 2.6181 & 2.5809 & 2.5682 & 2.5559 & 2.5731 & 2.5582 \\
			&WLRTF & 1.0560 & 1.0437 & 1.0288 & 1.0299 & 1.0218 & 1.0274 & 2.6009 & 2.5706 & 2.5642 & 2.5566 & 2.5571 & 2.5491 \\
            & PLMF & 2.6133 & 2.5932 & 2.4054 & 2.2097 & 2.1266 & 2.0247 & 4.5582 & 4.3536 & 4.3294 & 4.1843 & 3.9818 & 3.8542 \\
            & TASR & 2.8188 & 2.7120 & 2.5591 & 2.1184 & 2.0066 & 1.8854 & 6.3872 & 6.1807 & 5.9552 & 5.0212 & 4.8000 & 4.5447 \\
			& \textbf{CTF} & \textbf{0.9215} & \textbf{0.8981} & \textbf{0.8890} & \textbf{0.8860} & \textbf{0.8766} & \textbf{0.8750} & \textbf{2.5865} & \textbf{2.5579} & \textbf{2.5548} & \textbf{2.5529} & \textbf{2.5517} & \textbf{2.5401} \\
			\hline \hline
			\multirow{5}{*}{Throughput} 
			& NNCP & 1.5079 & 1.4342 & 1.4287 & 1.3761 & 1.3708 & 1.3761 & 4.9207 & 4.7019 & 4.6404 & 4.5080 & 4.4484 & 4.4968 \\
			& BNLFT & 1.4241 & 1.3935 & 1.3791 & 1.3856 & 1.3695 & 1.3613 & 4.6031 & 4.4685 & 4.4537 & 4.4128 & 4.3595 & 4.3493 \\
			& WLRTF & 2.9576 & 2.9568 & 2.9564 & 2.9561 & 2.9562 & 2.9537 & 4.9161 & 4.9160 & 4.9165 & 4.9153 & 4.9159 & 4.9095 \\
            & PLMF & 2.4712 & 2.2602 & 2.4459 & 2.3328 & 2.4655 & 2.2329 & 3.6705 & 3.8363 & 3.8455 & 3.8209 & 3.7541 & 3.5119 \\
            & TASR & 4.3265 & 3.6419 & 3.4736 & 2.8803 & 2.8258 & 2.7417 & 5.9152 & 5.1844 & 5.0034 & 4.3744 & 4.3142 & 4.2709 \\
			& \textbf{CTF} & \textbf{1.3567} & \textbf{1.1945} & \textbf{1.1225} & \textbf{0.9907} & \textbf{0.9889} & \textbf{0.9782} & \textbf{3.0436} & \textbf{2.9225} & \textbf{2.8576} & \textbf{2.7732}  & \textbf{2.6978}  & \textbf{2.6178} \\
			\hline
		\end{tabular}
		\label{tab12}
	\end{center}
\end{table*}

\begin{table*}[t!]
\small
	\caption{Performance Comparison with Different Outlier Ratios on Dynamic Dataset (Best Results in Bold Numbers)}
	\vspace*{-0.3cm}
	\begin{center}
        \setlength{\tabcolsep}{1.55mm}
		\begin{tabular}{c|c|c c c c c c|c c c c c c}
			\hline
			\multicolumn{1}{c|}{\multirow{2}*{\textbf{QoS Attributes}}} & \multicolumn{1}{c|}{\multirow{2}*{\textbf{Methods}}} & \multicolumn{6}{c|}{\textbf{MAE}} & \multicolumn{6}{c}{\textbf{RMSE}} \\
			\cline{3-14}
			&  & $2\%$ & $4\%$ & $6\%$ & $8\%$ & $10\%$ & $20\%$ & $2\%$ & $4\%$ & $6\%$ & $8\%$ & $10\%$ & $20\%$ \\
			\hline \hline
			\multirow{5}{*}{Response Time}
			& NNCP & 1.1846 & 1.1451 & 1.1069 & 1.0692 & 1.0521 & 1.0204 & 2.6740 & 2.6393 & 2.6023 & 2.5826 & 2.5805 & 2.5799 \\
			& BNLFT & 1.1647 & 1.1253 & 1.0871 & 1.0654 & 1.0475 & 0.9936 & 2.6499 & 2.6149 & 2.5771 & 2.5687 & 2.5659 & 2.5646 \\ 
			& WLRTF & 1.1436 & 1.1036 & 1.0562 & 1.0435 & 1.0261 & 0.9758 & 2.6438 & 2.6079 & 2.5696 & 2.5609 & 2.5584 & 2.5581 \\
            & PLMF & 2.6379 & 2.6011 & 2.5798 & 2.4241 & 2.3315 & 2.3162 & 5.2828 & 5.0571 & 4.7704 & 4.5917 & 4.2765 & 4.0729 \\
            & TASR & 2.5125 & 2.4326 & 2.3589 & 2.3363 & 2.3292 & 2.3019 & 5.4851 & 5.4510 & 5.4229 & 5.4018 & 5.3942 & 5.3877 \\
			& \textbf{CTF} & \textbf{1.0292} & \textbf{0.9813} & \textbf{0.9357} & \textbf{0.9105} & \textbf{0.8879} & \textbf{0.8448} & \textbf{2.6369} & \textbf{2.6015} & \textbf{2.5627} & \textbf{2.5564} & \textbf{2.5541} & \textbf{2.5503} \\
			\hline \hline
			\multirow{5}{*}{Throughput} 
			& NNCP & 2.4853 & 2.0339 & 1.7508 & 1.5419 & 1.3926 & 1.0231 & 9.8925 & 7.6471 & 6.2757 & 5.2096 & 4.5026 & 2.8916 \\
			& BNLFT & 2.4335 & 1.9909 & 1.7163 & 1.5137 & 1.3693 & 1.0117 & 9.7267 & 7.4979 & 6.1664 & 5.1236 & 4.4319 & 2.8376 \\ 
			& WLRTF & 6.4309 & 4.8846 & 3.9224 & 3.3382 & 2.9562 & 2.0911 & 17.9461 & 11.6611 & 7.9882 & 5.9178 & 4.9156 & 3.1509 \\
            & PLMF & 5.3105 & 4.1556 & 3.0467 & 2.9632 & 2.3924 & 2.1807 & 13.7985 & 8.7995 & 6.6349 & 4.9651 & 3.8347 & 3.7906 \\
            & TASR & 5.7661 & 4.5595 & 3.8264 & 3.3965 & 3.1322 & 2.6317 & 14.8241 & 9.7143 & 6.8450 & 5.2958 & 4.6089 & 3.5886 \\
			& \textbf{CTF} & \textbf{2.2624} & \textbf{1.6385} & \textbf{1.3421} & \textbf{1.1437} & \textbf{1.0193} & \textbf{0.7323} & \textbf{9.5370} & \textbf{6.0759} & \textbf{4.3650} & \textbf{3.2415} & \textbf{2.7989} & \textbf{1.8020} \\
			\hline
		\end{tabular}
		\label{tab22}
	\end{center}
\end{table*}

\subsubsection{Impact of Dimensionality}

The parameter dimensionality $l$ controls the dimension of latent features in the factor matrices. To study the impact of $l$, we vary its value from 10 to 80 with a step size of 10. In this experiment, the training ratio is fixed at 0.5 and the outlier ratio (denoted as $o$) is set to 0.02 and 0.1. The results are illustrated in Figure~\ref{fig:dl}. As we can see, both MAE and RMSE take smaller values when dimensionality $l$ grows. This is because when $l$ takes larger values, more features of users and Web services will be captured, thus resulting in more accurate prediction results. We also observe on the throughput dataset that the performance tends to be stable when $l \geq 40$, which indicates that $l=40$ is sufficient for the factor matrices to approximate the original matrix well.

\subsubsection{Impact of Parameter $\gamma$}

Recall that $\gamma$ denotes the constant in the Cauchy loss. Here we study its impact on the performance of our method by varying its value in the range of $\{0.1, 0.5, 1, 5, 10, 20, 50\}$. In this experiment, the training ratio is fixed at 0.5, and the outlier ratio $o$ is set to 0.02 and 0.1 as well. The results are illustrated in Figure~\ref{fig:gamma}. As can be seen, our method is sensitive to $\gamma$. This is due to that $\gamma$ implicitly determines which data will be treated as outliers during the training phase. Thus we need to choose a proper $\gamma$ to achieve the best performance. As shown in Figure~\ref{fig:gamma}, $\gamma$ should take value around 1 on the response time dataset and around 20 on the throughput dataset to obtain accurate prediction results.

\subsubsection{Impact of Data Sparsity}

To evaluate the performance of our method comprehensively, it is also necessary to investigate the impact of the sparsity of training data. To this end, we vary the training ratio from 0.1 to 0.9 with a step size of 0.1. Apparently, different training ratio implies different data sparsity. In this experiment, we also set the outlier ratio $o$ to 0.02 and 0.1. The results are reported in Figure~\ref{fig:ratio}. From  Figure~\ref{fig:ratio}, we see that as the training ratio increases (i.e., the sparsity of data decreases), more accurate results are obtained. 


\subsection{Experiments for Time-Aware QoS Prediction}


\subsubsection{Parameter Settings}

In the experiments, we tune the parameters of all baseline methods following the guidance of the original papers. As for our method CTF, on the response time dataset, the parameters are set as $l=15$ and $\lambda_u = \lambda_s = \lambda_t = 0.1$. $\gamma$ is set to 10 when calculating MAE and 35 when calculating RMSE.  On the throughput dataset, the parameters are set as $l=15$ and $\lambda_u = \lambda_s = \lambda_t = 100$. $\gamma$ is fixed at 5 for both MAE and RMSE. As for NNCP, BNLFT, WLRTF and PLMF, the feature dimensionality is also set to 15.

\subsubsection{Experimental Results}

We first report the results by varying the training ratios in the range of $\{0.1, 0.2, 0.3, 0.7, 0.8, 0.9\}$ and fixing the outlier ratio at 0.1. The results are presented in Table~\ref{tab12}. From Table~\ref{tab12}, we can see that our method consistently shows better performance than all baseline methods on both datasets. The results verify the robustness of our method in the time-aware extension.

We then report the results by varying the outlier ratios in the range of $\{0.02, 0.04, 0.06, 0.08, 0.1, 0.2\}$ and fixing the training ratio at 0.5. The results are shown in Table ~\ref{tab22}, from which we observe that our method achieves better performance under different outlier ratios, which is similar to the results on the static dataset.

\subsection{Efficiency Analysis}
Here, we further investigate the runtime efficiency of our method. In this experiment, we fix the training ratio at 0.5 and the outlier ratio at 0.1. The runtime of different methods on the response time dataset is reported in Figure~\ref{fig:time}. On the static dataset, we can observe that CMF is very efficient. Its runtime is comparable to that of MF$_2$ and MF$_1$. It also runs much faster than CAP, TAP and DALF. On the dynamic dataset, although CTF runs slower than PLMF and TASR, it is faster than BNLFT and WLRTF and is comparable to NNCP.

\begin{figure}[t!]
	\centering
	\subfigure[{Static}]{
		\includegraphics[width=0.4815\columnwidth]{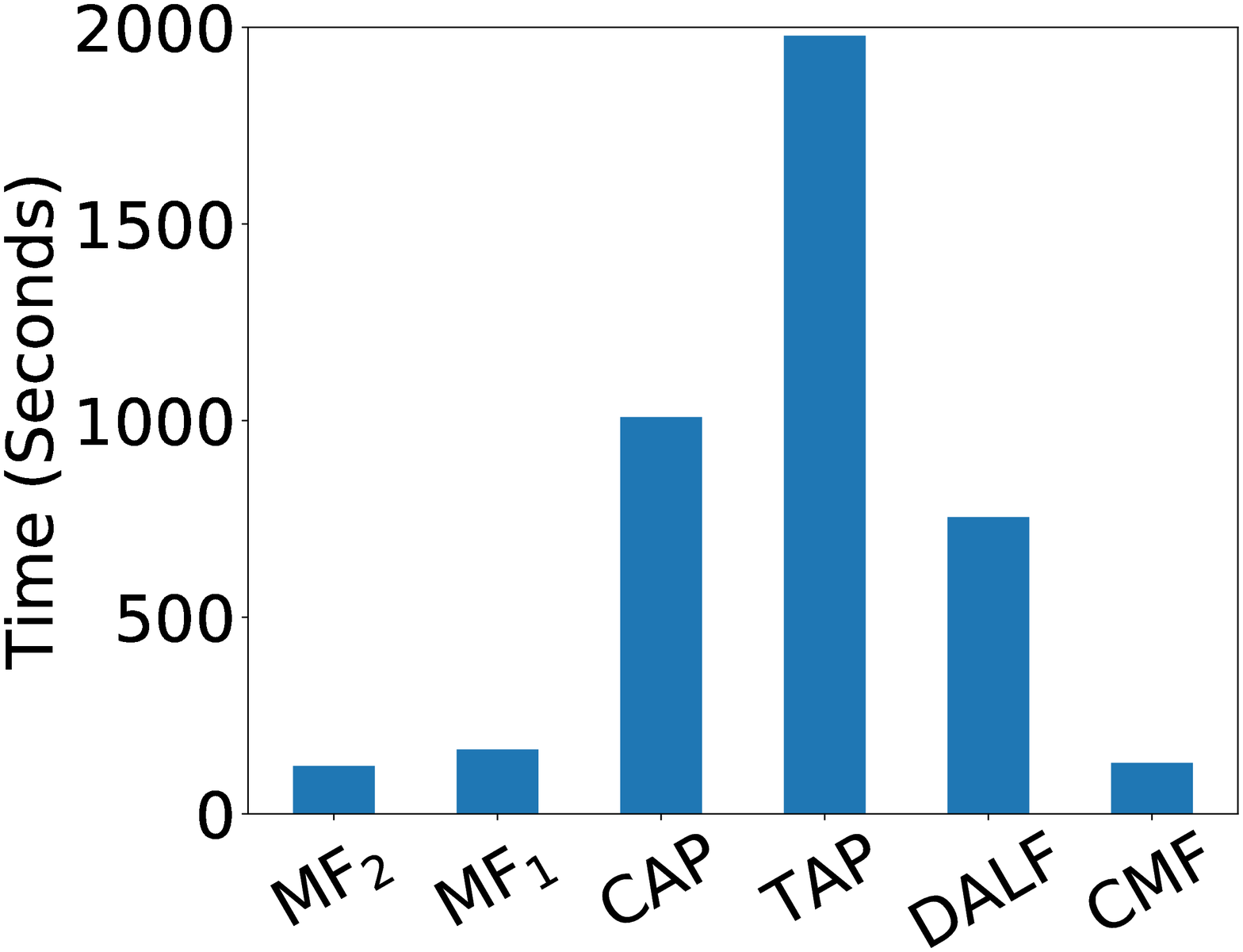}
	}
	\subfigure[{Dynamic}]{
		\includegraphics[width=0.470\columnwidth]{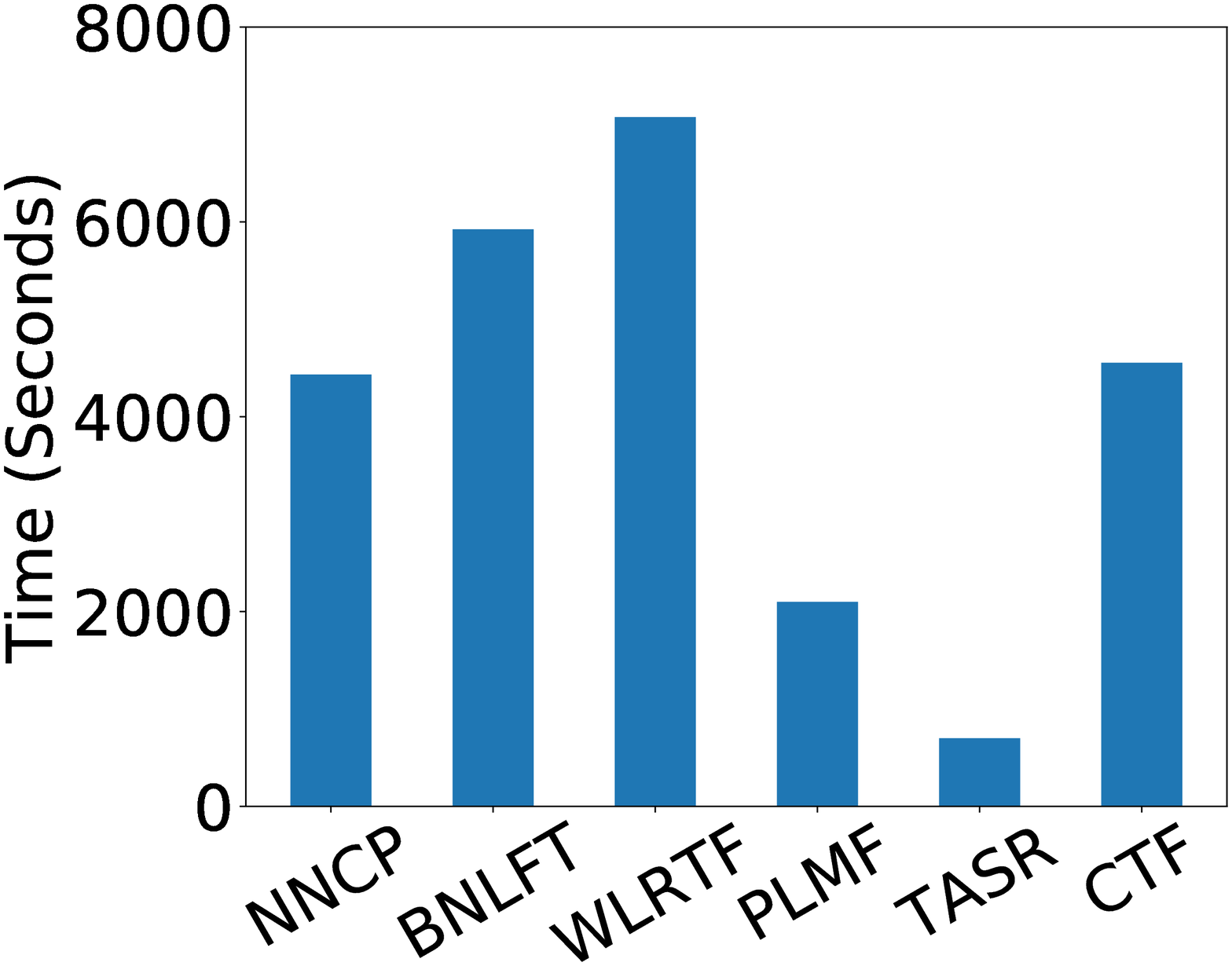}
	}
	\vspace*{-0.4cm}
	\caption{Runtime comparison on the response time dataset.}
	\label{fig:time}
	\vspace*{-0.3cm}
\end{figure}

\section{Related Work} \label{sec:rw}




\subsection{Collaborative QoS Prediction}


Most existing QoS prediction methods fall into collaborative filtering methods \cite{wu2015qos}, which can be further divided into two categories: memory-based methods \cite{zheng2011qos, cao2013hybrid, jiang2011effective, breese1998empirical, sarwar2001item} and model-based methods \cite{wu2018multiple, rangarajan2018qos, yang2018location, xie2016asymmetric, su2016web}. Memory-based methods predict unknown QoS values by employing the neighbourhood information of similar users and similar Web services \cite{zheng2010collaborative}, which further leads to user-based methods \cite{breese1998empirical}, service-based methods \cite{sarwar2001item} and hybrid methods \cite{zheng2011qos, cao2013hybrid, jiang2011effective} that systematically combine the user-based methods and service-based methods. Memory-based methods usually suffer from the data sparsity problem \cite{zhang2019covering, zheng2013personalized}, due to the limited number of Web services a single user will invoke. Model-based methods can deal with the problem of data sparsity, thus they have gained the most popularity \cite{zhang2019covering}. Model-based methods usually train a predefined prediction model based on existing QoS observations and then predict missing QoS values. For example, Wu et al. \cite{wu2017embedding} propose to train a factorization machine model for QoS prediction. Luo et al. \cite{luo2015web} introduce fuzzy neural networks and adaptive dynamic programming to predict QoS values. Matrix factorization is also a model-based technique and it has obtained the most attention \cite{zheng2013personalized, lo2015efficient, xie2016asymmetric, su2016web, zhang2019covering}. MF-based methods factorize the user-service matrix into two low-rank factor matrices with one factor matrix capturing the latent representations of users and another revealing the latent representations of Web services. Therefore, MF-based methods are able to automatically model the contributions to a specific QoS value from the user side and service side simultaneously, which usually results in better prediction performance. In addition, MF-based methods possess high flexibility of incorporating side information such as location \cite{he2014location}, contexts \cite{wu2018collaborative, wu2018multiple} and privacy \cite{liu2018differential}. MF-based methods can also be easily generalized for time-aware QoS prediction under the tensor factorization framework \cite{zhang2011wspred, zhang2014temporal, luo2019temporal, wang2019momentum}. There are also a few other kinds of time-aware QoS prediction methods like time series model-based methods \cite{ding2018time, amin2012automated} and neural networks-based methods \cite{xiong2018personalized}.

\subsection{Reliable QoS Prediction}

Although there are various QoS prediction methods,  few of them have taken outliers into consideration. However, as analyzed in Section~\ref{sec:mot}, some QoS observations indeed should be treated as outliers. Thus, the performance of existing methods may not be reliable. For example, most existing MF-based QoS prediction methods directly utilize $L_2$-norm to measure the discrepancy between the observed QoS values and the predicted ones \cite{zheng2013personalized, lo2015efficient, xu2016web, xie2016asymmetric, su2016web, wu2018collaborative}. It is widely accepted that $L_2$-norm is not robust to outliers \cite{cao2015low, xu2019adaptive, zhu2017robust}. As a consequence, the performance of MF-based methods may be severely influenced when QoS observations contain outliers.

In order to obtain reliable QoS prediction results, it is necessary to take outliers into consideration. One popular method to reduce the effects of outliers is replacing $L_2$-norm with $L_1$-norm because $L_1$-norm is more robust to outliers \cite{ke2005robust, eriksson2010efficient, zheng2012practical, meng2013cyclic}. For example, an $L_1$-norm low-rank MF-based QoS prediction method is introduced in \cite{zhu2018similarity}.  However, $L_1$-norm-based objective function is non-smooth and thus much harder to optimize \cite{xu2012sparse, meng2013robust}. Besides, $L_1$-norm is still sensitive to outliers, especially when outliers are far beyond the normal range of QoS values \cite{ding2017l1, xu2015multi}.

Another line of reliable QoS prediction is detecting outliers explicitly based on clustering algorithms. In \cite{wu2015qos}, Wu et al. propose a credibility-aware QoS prediction method, which employs a two-phase $k$-means clustering algorithm to identify untrustworthy users (i.e., outliers). Su et al. \cite{su2017tap} propose a trust-aware QoS prediction method, which provides reliable QoS prediction results via calculating the reputation of users by a beta reputation system and identifies outliers based on $k$-means clustering as well. In \cite{wu2019data}, a data-aware latent factor model is introduced, which utilizes the density peaks-based clustering algorithm \cite{rodriguez2014clustering} to detect unreliable QoS values.  However, it is difficult to choose a proper number of clusters, thus either some outliers may not be eliminated successfully or some normal values may be selected as outliers falsely. Our method does not detect outliers explicitly. Therefore it will not suffer from the misclassification issue.

\section{Conclusion} \label{sec:con}

In this paper, we have proposed a novel robust QoS prediction method, which utilizes Cauchy loss to measure the discrepancy between the observed QoS values and the predicted ones. Owing to the robustness of Cauchy loss, our method is resilient to outliers. That is, there is no need to detect outliers explicitly. Therefore, our method will not suffer from the problem of misclassification. Considering that the QoS performance may change over time, we have further extended our method to make it suitable for time-aware QoS prediction. To evaluate the efficiency and effectiveness of our method, we have conducted extensive experiments on both static and dynamic datasets. Experimental results have demonstrated that our method can achieve better performance than existing methods.


\bibliographystyle{ACM-Reference-Format}
\bibliography{CMFCTF}


\end{document}